\documentclass[lettersize,journal,onecolumn]{IEEEtran}
\usepackage{amsfonts}
\usepackage{algorithmic}
\usepackage{array}
\usepackage[caption=false,font=normalsize,labelfont=sf,textfont=sf]{subfig}
\usepackage{textcomp}
\usepackage{stfloats}
\usepackage{url}
\usepackage{verbatim}
\usepackage{graphicx}
\def\BibTeX{{\rm B\kern-.05em{\sc i\kern-.025em b}\kern-.08em
		T\kern-.1667em\lower.7ex\hbox{E}\kern-.125emX}}
\usepackage{balance}

\usepackage{cite}
\usepackage[cmex10]{amsmath}
\usepackage{fixltx2e}

\usepackage{times,color,multirow,rotating,bbm,latexsym,microtype,xcolor,amsthm}
\usepackage[shortlabels]{enumitem}
\usepackage{quoting}

\quotingsetup{vskip=0pt,leftmargin=0pt}
\theoremstyle{definition}
\newtheorem{theorem}{Theorem}
\newtheorem{lemma}{Lemma}
\newtheorem{remark}{Remark}
\newtheorem{example}{Example}

\newcommand{\parastretch}{\emergencystretch 3em}

\AtBeginDocument{
	\newcommand{\figurewidth}{0.5\columnwidth}  
}

\renewcommand{\figurename}{Figure}
\newcommand{\E}[1]{\mathbb{E}\left[#1\right]}
\newcommand{\Var}[1]{\mathbf{Var}\left[#1\right]}
\newcommand{\Cov}[1]{\mathbf{Cov}\left[#1\right]}

\newcommand{\indicator}[1]{{\mathbf{1}}\left( #1 \right)}
\newcommand{\assumption}[1]{{\rm \boldsymbol{A#1}}}
\DeclareMathOperator*{\argmax}{arg\,max}

\usepackage[absolute,showboxes]{textpos}
\TPMargin{5pt}
\newcommand{\copyrightstatement}{
	\begin{textblock}{0.828}(0.088,0.95)    
		\noindent
		\footnotesize
		\textcopyright This work has been submitted to the IEEE for possible publication. 
		Copyright may be transferred without notice, after which this version may no 
		longer be accessible.
	\end{textblock}
}

\begin{document}
\copyrightstatement
	
\title{
	Bayes Risk Consistency of Nonparametric Classification Rules 
	for Spike Trains Data
}
\author{
	Miros\l{}aw~Pawlak,~\IEEEmembership{Member,~IEEE,}
	Mateusz~Pabian,~\IEEEmembership{Student Member,~IEEE,}
	and~Dominik~Rzepka,~\IEEEmembership{Member,~IEEE}%
	\thanks{
		M.~Pawlak is with the Department of Electrical and Computer Engineering, 
		University of Manitoba, R3T 5V6 Winnipeg, Canada,
		and with the Department of Measurement and Electronics, 
		AGH University of Krakow, 30-059 Kraków, Poland.
		E-mail: Miroslaw.Pawlak@umanitoba.ca.
	}%
	\thanks{
		M.~Pabian and D.~Rzepka are with the Department of Measurement and Electronics, 
		AGH University of Krakow, 30-059 Kraków, Poland.
	}%
	\thanks{
		A preliminary version of this paper was presented at IEEE ICASSP 2023.
	}%
}

\markboth{Preprint}
{
	M.~Pawlak, M.~Pabian, D.~Rzepka: 
		Bayes Risk Consistency of Nonparametric Classification Rules 
		for Spike Trains Data
}

\maketitle

\begin{abstract}
	\noindent Spike trains data find a growing list of applications in computational 
	neuroscience,
	imaging, streaming data and finance. Machine learning strategies for spike
	trains are based on various neural network and probabilistic models. The probabilistic
	approach is relying on parametric or nonparametric specifications of the underlying spike
	generation model. In this paper we
	consider the two-class statistical classification problem for a class of spike train data
	characterized by nonparametrically specified intensity functions. We derive the optimal
	Bayes rule and next form the plug-in nonparametric kernel classifier. Asymptotical
	properties of the rules are established including the limit with respect to the
	increasing recording time interval and the size of a training set.
	In particular the convergence of the kernel classifier to the Bayes rule is proved.
	The obtained results are supported by a finite sample simulation studies.
\end{abstract}


\section{Introduction}
\label{sec:01_introduction}

\IEEEPARstart{E}{vent} driven systems are often encountered in science and engineering. 
In such systems 
data are represented by point processes that define arrival times of events. In 
computational neuroscience and machine learning this type of data are called spike 
trains~\cite{gerstner2014neuronal, jang2019introduction, shchur2021neural}. 
In optical communication systems one observes a train of impulses (representing
a point process) emitted by photon-sensitive detectors. Signal detection and estimation
methods for such the so-called Poisson regime channels have been extensively examined
in communication and information theory~\cite{david69, guo08, merhav021}.
On the other hand, the mathematical theory of point processes has been extensively
studied in the statistical and stochastic processes
literature~\cite{daley2003introduction, andersen2012statistical}.
However, the research on event type processes from the statistical 
classification theory~\cite{devroye2013probabilistic} perspective has been initiated 
very recently~\cite{cholaquidis2017classification, rong2021error}. Probabilistic 
spiking neural networks have been introduced for supervised and unsupervised learning 
problems~\cite{jang2019introduction}. Various simulation results have been reported 
supporting their usefulness without, however, any accuracy studies and fundamental 
limits.

In this paper, we develop the Bayes strategy~\cite{devroye2013probabilistic} for the 
spiking data supervised classification problem. This strategy can be applied to 
research problems where event ocurrence is the primary information 
carrier~\cite{mazza2019rtbust, fanaee2014event}. We consider a class of 
temporal spiking processes that are characterized by non-random intensity functions. 
The intensity function plays the central role in our theory as it describes the
local rate of occurrence of spikes. For such processes 
(Section~\ref{sec:02_bayes_rule}) we derive the optimal Bayes rule in terms of class 
intensity functions. In Section~\ref{sec:03_bayes_bounds} the limit behavior the Bayes 
rule with respect to the increasing length of the observation interval is examined. In 
Section~\ref{sec:04_nonparametric} the plug-in nonparametric kernel classification rule 
from multiple replications of spiking processes is proposed. This is followed by the 
asymptotical optimality result, i.e., the convergence of the kernel rule to the Bayes 
rule. This result can be considered as the counterpart of the result 
in~\cite{greblicki1978asymptotically} concerning the classical plug-in nonparametric 
classification rules defined in the finite-dimensional Euclidean space. The spike train 
data are characterized by the variable-length continuous-discrete vectors of event times 
and their number over a given observation interval. The main mathematical tool in our
asymptotic analysis is the theory of the martingale decomposition for counting
processes~\cite{andersen2012statistical}.

It is also worth mentioning that the asymptotic optimality does not hold if one observes 
the long single realization of the underlying spiking process. In fact, the intensity 
estimation problem for spiking processes does not fall into the classical 
large-sample-smaller distance between sample points framework as the point process is 
casual in time~\cite{diggle1988equivalence}. Hence, for a fixed observation interval 
one must increase the number of events. This can be achieved by either scaling the 
intensity function or by using the replicates of the spiking process. The former 
approach can be based on the multiplicative intensity model due to 
Aalen~\cite{aalen1978nonparametric}, whereas the latter one (used in this paper) is 
the standard machine learning strategy, where the replicates form the training set. In 
this case the resulting kernel estimate will be obtained by aggregating kernel 
estimates from single realizations. Our asymptotic results are supported by 
simulation studies presented in Section~\ref{sec:05_simulations}. The preliminary
version of the results developed in this paper has been reported
in~\cite{pawlak2023asymptotically}.

The symbol $\indicator{A}$ denotes the indicator function of the set $A$. We shall use 
the notation $(P)$ for the convergence in probability, whereas $(a.s.)$ denotes the 
convergence with probability one. Also $a_{T} \prec cb_{T}$ denotes the asymptotic 
bound, i.e., $a_{T} \leq cb_{T}$ for sufficiently large $T$. Furthermore, 
$\overline{\lim}$, $\underline{\lim}$ denote the limit superior and inferior,
respectively.  Also, by $M_f$ we will denote the Lipschitz constant of a function
$f(t)$ , i.e., $f(t)$ meets the Lipschitz condition if
$|f(t_1) - f(t_2)| \leq M_f |t_1-t_2|$ for all $t_1,t_2$.

\section{Bayes Classification Rule}
\label{sec:02_bayes_rule}

\noindent A temporal spiking process~\mbox{$\{N(t),t\geq0\}$} consists of a sequence 
of random times $\{t_{i}\}$ of isolated events in time such that~\mbox{$N(0)=0$}. The 
process~$N(t)$ can be defined by the counting function \mbox{
	$N(t)=\sum_i \indicator{t_i \leq t}$
} which is the number of events in \mbox{$\left[0,t\right]$}. 
We assume that the process is observed on 
the time window~$[0,T]$ and is characterized by the non-random intensity 
function~$\lambda(t)$ that is defined for all $t\geq0$. 
This is a non-negative function that describes the local 
arriving rate of events such that \mbox{$\E{N(T)}=\int_{0}^{T}\lambda(u)du$} 
is the average number of events in~$[0,T]$. Hence, the observed on~$[0,T]$ 
process~$N(t)$ can be represented by the variable-length vector 
\mbox{$\mathbf{X}=[t_{1},\ldots,t_{N};N]$}, where \mbox{$0<t_{1}<\cdots<t_{N}<T$} are 
the event times and~\mbox{$N=N(T)$}. Writing $\left[t_1,\ldots,t_N;N\right]$ we 
emphasize the fact that the data vector consists of two parts: the occurrence times 
\mbox{$\lbrace t_1,\ldots,t_N \rbrace$} and~$N$ being the number of events in 
\mbox{$\left[0,T\right]$}. The former is the continuous part of the vector 
$\mathbf{X}$, whereas the latter is its discrete part.

The goal of this paper is to develop a rigorous classification methodology for the 
aforementioned class of spiking processes based on the Bayes theory of 
classification~\cite{devroye2013probabilistic}. Without a loss of generality we 
consider a two-class classification problem (see Section~\ref{sec:06_conclusions} for 
the generalization to the multi-class case) where class labels are denoted as 
\mbox{$\omega_{1}$, $\omega_{2}$} with the priori probabilities \mbox{
	$\pi_{1}$, $\pi_{2}$
},  respectively. In order to form the optimal Bayes rule we recall the following known 
result~\cite{daley2003introduction} on the joint 
occurrence density of~$\mathbf{X}$
\begin{equation}
	f(\mathbf{x})=\prod_{i=1}^{N}\lambda(t_{i})
	{\rm exp}\left(-\int_{0}^{T}\lambda(u)du\right) 
	\label{eqn:2.1}
\end{equation}
for $N=N(T)\geq1$, whereas if $N=0$ then \mbox{
	$f(\mathbf{x})=\exp\left(-\int_{0}^{T}\lambda(u)du\right)$
}. It is worth noting that~\eqref{eqn:2.1} is the continuous-discrete distribution and 
by virtue of~\eqref{eqn:2.1} the marginal density of the 
occurrence times \mbox{$\lbrace t_1,\ldots,t_N \rbrace$} for \mbox{
	$N \in \lbrace 0,1,\ldots \rbrace$
} is given by
\begin{equation}
	\begin{split}
		\sum_{n=0}^{\infty} & f\left(t_1,\ldots,t_N; N=n\right) =
		\exp\left(-\int_{0}^{T}\lambda(u)du\right) \\
		& + \exp\left(-\int_{0}^{T}\lambda(u)du\right)
		\sum_{n=1}^{\infty} \prod_{j=1}^{n} \lambda(t_j)
	\end{split}
	\,,
	\label{eqn:2.2}
\end{equation}
which is defined over the simplex regions \mbox{$
	\mathbb{C}_n=\lbrace \left(t_1,\ldots,t_n \right):
	0 \leq t_1 \le \ldots \le t_n \leq T \rbrace
	$}, \mbox{$n=1,2,\ldots$}. The formula in~\eqref{eqn:2.2} defines the proper density 
over $\{\mathbb{C}_n\}$, i.e., we have
\begin{equation}
	\begin{split}
		& \exp\left(-\int_{0}^{T}\lambda(u)du\right) \\
		& +
		\exp\left(-\int_{0}^{T}\lambda(u)du\right)
		\sum_{n=1}^{\infty} \int_{\mathbb{C}_n} \prod_{j=1}^{n} \lambda(t_j) 
		dt_1 \cdots dt_n
		= 1
	\end{split}
	\label{eqn:2.2a}\,.
\end{equation}

In the context of the classification problem the class occurrence densities 
in~\eqref{eqn:2.1} will be denoted~$f_1(\mathbf{x})$ and~$f_2(\mathbf{x})$ depending 
whether~$\mathbf{X}$ comes from class~$\omega_{1}$ (denoted as \mbox{
	$\mathbf{X} \in \omega_{1}$
}) or if \mbox{$\mathbf{X} \in \omega_{2}$}, respectively. The corresponding class
intensities are \mbox{$\lambda_{1}(t)$, $\lambda_{2}(t)$} being the non-negative 
functions defined on \mbox{$\left[0, \infty\right)$}. Then using~\eqref{eqn:2.1}, 
one can form the optimal Bayes rule \mbox{
	$\psi_{T}^{*}$: $\mathbf{X}\in\omega_{1}$
} if 
\begin{equation}
	\prod_{i=1}^{N}\frac{\lambda_{1}(t_{i})}{\lambda_{2}(t_{i})}
	{\rm exp}\left(\int_{0}^{T}\left[\lambda_{2}(u)-\lambda_{1}(u)\right]du\right)
	\geq\frac{\pi_{2}}{\pi_{1}}
	\label{eqn:2.3} \,.
\end{equation}
assuming that~$N\geq1$ and 
$
\exp\left(\int_{0}^{T}\left[\lambda_{2}(u)-\lambda_{1}(u)\right]du\right) 
\geq \frac{\pi_{2}}{\pi_{1}}
$ if~$N=0$. Clearly, if the reverse inequality in~\eqref{eqn:2.3} holds, then we 
classify $\mathbf{X}$ to $\omega_{2}$. The log transform of~\eqref{eqn:2.3} gives the 
alternative convenient form 
of the rule~$\psi_{T}^{*}$, i.e., \mbox{$\mathbf{X}\in\omega_{1}$} if 
\begin{equation}
	\sum_{i=1}^{N}\log
	\left(\frac{\lambda_{1}(t_{i})}{\lambda_{2}(t_{i})}\right)
	\geq\gamma
	\label{eqn:2.4} \,,
\end{equation}
where \mbox{$
	\gamma=
	\int_{0}^{T}\left[\lambda_{1}(u)-\lambda_{2}(u)\right]du +
	\log\left(\frac{\pi_{2}}{\pi_{1}}\right)
	$}. The rule in~\eqref{eqn:2.4} can be usefully written in terms of the stochastic 
integral of the log-ratio $\log\left(\frac{\lambda_{1}(t)}{\lambda_{2}(t)}\right)$ 
with respect to the increments of the counting process $N(t)$, i.e., 
\mbox{$\mathbf{X}\in\omega_{1}$} if 
\begin{equation}
	\int_{0}^{T}
	\log
	\left(\frac{\lambda_{1}(t)}{\lambda_{2}(t)}\right)
	dN(t)
	\geq\gamma
	\label{eqn:2.4a} \,.
\end{equation}
Here $N(t)$ is the aforemenioned counting process with the intensity function 
$\lambda(t)$, where
\begin{equation}
	\lambda(t) = 
	\begin{cases}
		\lambda_{1}(t) & \text{if} \enspace \mathbf{X}\in\omega_{1} \\
		\lambda_{2}(t) & \text{if} \enspace \mathbf{X}\in\omega_{2}
	\end{cases}
	\label{eqn:2.4b} \,.
\end{equation}

For our further considerations it is useful to represent the class intensity functions 
on $[0,T]$ in terms of the so-called intensity factor and shape 
function~\cite{gajardo2021cox}. 
Thus, let \mbox{$\lambda_{1}(t)=\tau_{1}p_{1}(t)$}, 
\mbox{$\lambda_{2}(t)=\tau_{2}p_{2}(t)$}, where 
\begin{equation}
	\tau_{i}=\int_{0}^{T}\lambda_{i}(u)du ,\enspace p_{i}(t)=\lambda_{i}(u)/\tau_{i} ,\enspace
	i=1,2
	\label{eqn:2.5} \,.
\end{equation}
Clearly~$p_{1}(t)$,~$p_{2}(t)$ are well-defined probability density functions 
on~$[0,T]$. The 
representation in~\eqref{eqn:2.5} allows us to represent the classification problem in 
terms of the class intensity factors and shape densities, and employ 
information-theoretic divergence measures. Using~\eqref{eqn:2.5}, we can 
rewrite the rule in~\eqref{eqn:2.4} as follows, \mbox{$\mathbf{X}\in\omega_{1}$} if 
\begin{equation}
	\sum_{i=1}^{N}\log\left(
	\frac{p_{1}(t_{i})}{p_{2}(t_{i})}
	\right)
	\geq\eta
	\label{eqn:2.6} \,,
\end{equation}
where \mbox{$
	\eta=
	\tau_{1} - \tau_{2} + 
	N\log\left(\frac{\tau_{2}}{\tau_{1}}\right) +
	\log\left(\frac{\pi_{2}}{\pi_{1}}\right)
	$}. 
\parastretch The Bayes rule~$\psi_{T}^{*}$ in~\eqref{eqn:2.6} will be written as 
\mbox{$W_T(\mathbf{X})\geq\eta_T$} emphasizing the fact that the vector~$\mathbf{X}$ is
observed within the time window~$[0,T]$. \\
It is worth noting that if \mbox{$\lambda_{1}(t)=\lambda_{1}$} and
\mbox{$\lambda_{2}(t)=\lambda_{2}$}, i.e., if we have the homogeneous spike train data 
then the Bayes rule takes the following form \mbox{
	$\psi_{T}^{*}$: $\mathbf{X}\in\omega_{1}$
}
\begin{equation}
	N \log \left(\frac{\lambda_{1}}{\lambda_{2}}\right) + 
	T \left(\lambda_{2}-\lambda_{1}\right) 
	\geq \log \left(\frac{\pi_{2}}{\pi_{1}}\right)
	\label{eqn:2.7} \,,
\end{equation}
provided that~$N\geq1$. In the case~$N=0$ this reads as
$
\lambda_{2}-\lambda_{1} \geq \frac{1}{T} \log \left(\frac{\pi_{2}}{\pi_{1}}\right)
$. The risk associated with the rule \mbox{$\psi_{T}^{*}(\mathbf{x})$} 
in~\eqref{eqn:2.4} (or~\eqref{eqn:2.6}) is defined as \mbox{
	$\mathbf{R}_{T}^{*}=\mathbf{P(\psi_{T}^{*}(\mathbf{X})\neq Y)}$	
} and is referred as the Bayes risk. Here \mbox{
	$Y\in \{\omega_{1}, \omega_{2}\}$
} is the true class label of $\mathbf{X}$. For our future studies we express the 
Bayes risk in terms of the decision function $W_T(\mathbf{X})$, i.e., we write
\begin{equation}
	\begin{split}
		\mathbf{R}_{T}^{*} = {}
		&\mathbf{P}\left(\mathbf{W}_{T}(\mathbf{X})\geq\eta_{T}|
		\mathbf{X}\in\omega_{2}\right)\pi_{2}\\
		&+ \mathbf{P}\left(\mathbf{W}_{T}(\mathbf{X})<\eta_{T}|
		\mathbf{X}\in\omega_{1}\right)\pi_{1}
	\end{split}
	\label{eqn:2.8} \,.
\end{equation}
It is an important question to evaluate the Bayes risk. This includes various bounds 
on~$\mathbf{R}_{T}^{*}$ and the behavior of~$\mathbf{R}_{T}^{*}$ as a function of~$T$. 
In Sections~\ref{subsec:03_bayes_bounds/decision} and~\ref{subsec:03_bayes_bounds/risk} 
we present results concerning such issues.

The presented results rely on the following local decomposition 
(see Appendix~A) of the increment~$dN(t)$ of the point process~$N(t)$. Hence, we have
\begin{equation}
	dN(t) = \lambda(t)dt + dM(t)
	\label{eqn:2.9} \,,
\end{equation}
where~$\lambda(t)$ is the intensity function of~$N(t)$, and~$dM(t)$ is a zero mean 
process with uncorrelated but non-stationary increments. The formula 
in~\eqref{eqn:2.9} can be viewed as the local signal plus noise decomposition, where 
the noise process~$dM(t)$ reveals the local martingale 
structure~\cite{andersen2012statistical}. Appendix~A gives the pertinent results
concerning the martingale decomposition of the underlying spiking process.

The decomposition in~\eqref{eqn:2.9} allows us to express the classification rule 
in~\eqref{eqn:2.4} (or its version in~\eqref{eqn:2.6}) in the convenient stochastic 
integral form. In fact, by virtue of~\eqref{eqn:2.9} and~\eqref{eqn:2.4a} we write the 
left-hand side of~\eqref{eqn:2.4a} as
\begin{equation}
	\begin{split}
		\int_{0}^{T} \log \left(\frac{\lambda_{1}(t)}{\lambda_{2}(t)}\right)dN(t) = {}
		& \int_{0}^{T} \log \left(\frac{\lambda_{1}(t)}{\lambda_{2}(t)}\right)\lambda(t)dt \\
		& + \int_{0}^{T} \log \left(\frac{\lambda_{1}(t)}{\lambda_{2}(t)}\right)dM(t)
	\end{split}
	\label{eqn:2.9a} \,,
\end{equation}
where $\lambda(t)$ is given in~\eqref{eqn:2.4b} and $M(t)$ is the corresponding noise
process defined in~\eqref{eqn:2.9}.
The first term in~\eqref{eqn:2.9a} is the bias term of the optimal decision function, 
whereas the second one is the zero mean random variable contributing to the 
statistical variability of the rule. In Section~\ref{sec:03_bayes_bounds} we show that 
the normalized version of this term converges exponentially fast to zero as 
$T \to \infty$ with probability one.

\section{The Bayes Rule and Risk: Bounds and Asymptotic Behavior}
\label{sec:03_bayes_bounds}

\subsection{The Bayes Decision Function}
\label{subsec:03_bayes_bounds/decision}

\noindent In this section we examine the optimal decision function derived 
in~\eqref{eqn:2.4} or 
its alternative form in~\eqref{eqn:2.6}. Owing to the decomposition 
in~\eqref{eqn:2.9a} and using~\eqref{eqn:2.4a} we can arrive to the following 
equivalent form of the rule $\psi_{T}^{*}$ in~\eqref{eqn:2.6}, 
$\mathbf{X}\in\omega_{1}$ if
\begin{equation}
	U_{T}(\mathbf{X}) \geq
	\alpha_{T} + \log\left(\frac{\pi_{2}}{\pi_{1}}\right)
	\label{eqn:3.1} \,,
\end{equation}
where
\begin{equation}
	U_{T}(\mathbf{X}) =
	\int_{0}^{T} g(t)dM(t)
	\label{eqn:3.2} \,.
\end{equation}
Here $
g(t) = 
\log\left(\frac{\lambda_{1}(t)}{\lambda_{2}(t)}\right) =
\log\left(\frac{p_{1}(t)}{p_{2}(t)}\right) + \log\left(\frac{\tau_{1}}{\tau_{2}}\right)
$ and
\begin{equation}
	\begin{split}
		\alpha_{T} = {}
		& \tau_{1} - \tau_{2} + \log\left(\frac{\tau_{2}}{\tau_{1}}\right)
		\int_{0}^{T} \lambda(t)dt \\
		&+ \int_{0}^{T} \log\left(\frac{p_{2}(t)}{p_{1}(t)}\right) \lambda(t)dt
	\end{split}
	\label{eqn:3.3} \,,
\end{equation}
where $\lambda(t)$ is specified in~\eqref{eqn:2.4b}.

It is worth noting that $U_T(\mathbf{X})$ in~\eqref{eqn:3.1} represents the stochastic 
part of the Bayes rule. This takes the form of the stochastic integral with respect to 
the increments of the martingale process $M(t)$. It is 
known~\cite{andersen2012statistical} that the martingale property is preserved under 
stochastic integration. Hence, since $\E{dM(t)}=0$ the process
\begin{equation*}
	\left\{U_t(\mathbf{X}) = \int_{0}^{t} g(u) dM(u), 0 \leq t \leq T \right\}
\end{equation*}
is a zero mean local martingale associated with the counting process $N(t)$, 
see Appendix~A for further details. In addition, the integral in~\eqref{eqn:3.2} is
specified by the log-ratio $\log\left(\frac{\lambda_{1}(t)}{\lambda_{2}(t)}\right)$ 
and this is generally the 
unbounded function. To prevent this singularity it suffices to assume the class 
intensities $\lambda_{1}(t)$, $\lambda_{2}(t)$ that are bounded away from zero. 
Moreover, intensity functions are commonly bounded. All these restrictions can be 
formalized by the following assumption that will be used in the paper. Hence, assume 
that there exist positive numbers $\delta$ and $C$ such that
\begin{flalign}
	\assumption{1:} 
	\enspace
	0 < \delta \leq \lambda_{i}(t) \leq C,
	\enspace
	i = 1,2,
	\enspace
	\text{for all}
	\enspace
	t\geq 0
	\label{eqn:3.3a} \,. &&
\end{flalign}
We refer to~\cite{birge1993rates, van1995exponential} for some weaker conditions for 
the existence of the aforementioned log-ratio.

In this section we present the preliminary results that characterize the Bayes rule
specified by~\eqref{eqn:3.2} and~\eqref{eqn:3.3}. This includes some bounds on the 
threshold $\alpha_{T}$ in~\eqref{eqn:3.3} and the statistical properties of the 
stochastic term in~\eqref{eqn:3.2}. To do so, we recall that the Kullback-Leibler~(KL) 
divergence~\cite{ghosal2017fundamentals} between densities $p(t)$ and $q(t)$ on $[0,T]$ 
is defined as follows
\begin{equation}
	\mathbf{K}_{T} \left(p \parallel q\right) = 
	\int_{0}^{T} \log\left(\frac{p(t)}{q(t)}\right)p(t)dt
	\label{eqn:3.4}\,.
\end{equation}
It is known that $\mathbf{K}_{T} \left(p \parallel q\right) \geq 0$ and 
$\mathbf{K}_{T} \left(p \parallel q\right) = 0$ if $p=q$.

The following lemma gives the upper and lower bounds for the threshold $\alpha_{T}$ 
in~\eqref{eqn:3.3} in terms of the KL divergence between the class densities and the 
normalized square distance between the corresponding intensity factors. We will find 
these bounds useful in evaluating the Bayes risk.
\begin{lemma}
	\label{lemma:1}
	Let $\alpha_{T}$ be the threshold defined in~\eqref{eqn:3.3}. Then we have
	\begin{enumerate}[(a)]
		\item If $\mathbf{X}\in\omega_{1}$ then
		\begin{equation}
			\begin{split}
				-\frac{\left(\tau_{1} - \tau_{2}\right)^2}{\tau_2}
				& -\tau_{1}\mathbf{K}_{T}\left(p_1 \parallel p_2\right)\\
				& \leq \alpha_{T} \leq
				-\tau_{1}\mathbf{K}_{T}\left(p_1 \parallel p_2\right)
			\end{split}
			\label{eqn:3.4a}\,.
		\end{equation}
		\item If $\mathbf{X}\in\omega_{2}$ then
		\begin{equation}
			\begin{split}
				\tau_{2}\mathbf{K}_{T}\left(p_2 \parallel p_1\right)
				& \leq \alpha_{T} \\
				& \leq \frac{\left(\tau_{1} - \tau_{2}\right)^2}{\tau_1}
				+ \tau_{2}\mathbf{K}_{T}\left(p_2 \parallel p_1\right)
			\end{split}
			\label{eqn:3.4b}\,,
		\end{equation}
	\end{enumerate}
	where $p_{1}, p_{2},\tau_{1},\tau_{2}$ are defined in~\eqref{eqn:2.5}. 
	The proof of Lemma~\ref{lemma:1} is given in Appendix~B.
\end{lemma}

As the KL divergence is non-negative, then Lemma~\ref{lemma:1}(a) yields 
$\alpha_{T}\leq0$ if $\mathbf{X}\in\omega_{1}$, whereas Lemma~\ref{lemma:1}(b) gives 
$\alpha_{T}\geq0$ for $\mathbf{X}\in\omega_{2}$. Also it is seen that $\alpha_{T}$ 
lies in the interval of the length $\left(\tau_{1}-\tau_{2}\right)^2/\tau_{2}$ and 
$\left(\tau_{1}-\tau_{2}\right)^2/\tau_{1}$ if $\mathbf{X}\in\omega_{1}$ and 
$\mathbf{X}\in\omega_{2}$, respectively. It is also worth noting that \mbox{
	$\left(\tau_{1}-\tau_{2}\right)^2=\left\{
	\int_{0}^{T}\left(\lambda_{1}(t)-\lambda_{2}(t)\right)dt
	\right\}^2$
} represents the square of the difference of the average number of events on $[0,T]$
coming from
classes $\omega_{1}$ and $\omega_{2}$. As a result, if $\tau_{1} = \tau_{2}$ then
$\alpha_{T} = -\tau_{1}\mathbf{K}_{T}\left(p_1 \parallel p_2\right)$
for $\mathbf{X}\in\omega_{1}$ and
$\alpha_{T} = \tau_{2}\mathbf{K}_{T}\left(p_2 \parallel p_1\right)$
for $\mathbf{X}\in\omega_{2}$.

The next result concerns the stochastic part $U_{T}(\mathbf{X})$ defined 
in~\eqref{eqn:3.2}. This is given in the form of the stochastic integral of the log-ratio 
$\log\left(\frac{\lambda_{1}(t)}{\lambda_{2}(t)}\right)$ with respect to the 
increments of $M(t)$ such that $\E{U_{T}(\mathbf{X})}=0$. In the following lemma we 
evaluate the basic statistical feature of this term by deriving its variance.
\begin{lemma}
	\label{lemma:2}
	Let us consider the stochastic part $U_T(\mathbf{X})$ of the Bayes rule 
	in~\eqref{eqn:3.1}. Then,
	\begin{enumerate}[(a)]
		\item If $\mathbf{X}\in\omega_{1}$ then
		\begin{equation}
			\begin{split}
				& \Var{U_{T}(\mathbf{X})} \\
				& =
				\tau_{1}\int_{0}^{T}\left\{
				\log\left(\frac{p_{1}(t)}{p_{2}(t)}\right)
				+ log\left(\frac{\tau_{1}}{\tau_{2}}\right)
				\right\}^2 p_{1}(t)dt
			\end{split}
			\label{eqn:3.5}\,.
		\end{equation}
		\item If $\mathbf{X}\in\omega_{2}$ then
		\begin{equation}
			\begin{split}
				& \Var{U_{T}(\mathbf{X})} \\
				& =
				\tau_{2}\int_{0}^{T}\left\{
				\log\left(\frac{p_{2}(t)}{p_{1}(t)}\right)
				+ log\left(\frac{\tau_{2}}{\tau_{1}}\right)
				\right\}^2 p_{2}(t)dt
			\end{split}
			\label{eqn:3.6}\,.
		\end{equation}
	\end{enumerate}
	The proof of Lemma~\ref{lemma:2} is given in Appendix~B.
\end{lemma}

The formulas in Lemma~\ref{lemma:2} can be expressed in terms of the higher-order KL 
divergence between two class densities referred to as the KL 
variation~\cite{ghosal2017fundamentals}. Hence, let
\begin{equation}
	\mathbf{V}_T \left(p \parallel q\right) =
	\int_{0}^{T} \log^2\left(\frac{p(t)}{q(t)}\right)p(t) dt
	\label{eqn:3.7}
\end{equation}
be the KL variation between densities $p(t)$ and $q(t)$ on $[0,T]$. Note that 
$\mathbf{V}_T \left(p \parallel q\right)=0$ if $p=q$. Moreover, the following result 
describes the relationship between $\mathbf{V}_T \left(p \parallel q\right)$ and the 
standard KL divergence in~\eqref{eqn:3.4}.
\begin{lemma}
	\label{lemma:3}
	For any pair of probability densities $p,q$ on $[0,T]$ we have
	\begin{equation}
		\mathbf{K}_T \left(p \parallel q\right) \leq
		\sqrt{\mathbf{V}_T \left(p \parallel q\right)}
		\label{eqn:3.8}\,.
	\end{equation}
	The bound in~\eqref{eqn:3.8} results from the direct application of the 
	Cauchy-Schwarz inequality. Returning back to the formula in~\eqref{eqn:3.5} we can 
	obtain that
	\begin{equation}
		\begin{split}
			\Var{U_{T}(\mathbf{X})} &=
			\tau_{1} \bigg\{
			\mathbf{V}_T \left(p_{1} \parallel p_{2}\right) 
			\bigg. \\
			& \left.
			+ 2\log\left(\frac{\tau_{1}}{\tau_{2}}\right)
			\mathbf{K}_T \left(p_{1} \parallel p_{2}\right) +
			\log^2\left(\frac{\tau_{1}}{\tau_{2}}\right)
			\right\}
		\end{split}
		\label{eqn:3.9}\,.
	\end{equation}
	The analogous formula can be written for~\eqref{eqn:3.6}.
\end{lemma}

It is an interesting question to examine the behavior of the stochastic term 
$U_{T}(\mathbf{X})$ for an increasing value of the observation interval $T$. In 
particular, we wish to derive an analog of the law of large numbers, i.e., the limit 
behavior of
\begin{equation}
	\frac{1}{T}U_{T}(\mathbf{X})=
	\frac{1}{T} \int_{0}^{T} g(t) dM(t)
	\label{eqn:3.9a}
\end{equation}
as $T\to\infty$, where $g(t)$ is the log-ratio 
$\log\left(\frac{\lambda_{1}(t)}{\lambda_{2}(t)}\right)$. To give an answer to such 
questions we need to put some condition on the growth of the assumed class of intensity 
functions. Hence, suppose that there exists positive number $d$ such that
\begin{flalign}
	\assumption{2:}
	\enspace
	\frac{1}{T} \int_{0}^{T} \lambda_{i}(u)du \to d,
	\enspace
	i=1,2
	\enspace
	\text{as}
	\enspace
	T\to\infty
	\label{eqn:3.10}\,. &&
\end{flalign}
The meaning of this condition is that the average number of events from the each class 
increases linearly with $T$. It is worth noting that for  intensity functions that are
integrable on $[0,\infty)$ the condition in~\eqref{eqn:3.10} holds with $d=0$. \\
Based on the assumption $\assumption{2}$ we wish to evaluate the limit behavior of
$\Var{U_{T}(\mathbf{X})}$ as $T\to\infty$. It is clear that such
limit may not exist. Nevertheless, using the assumption $\assumption{1}$ we can find the 
upper and lower bounds for $\Var{U_{T}(\mathbf{X})}$. In fact, recalling~\eqref{eqn:3.5} 
and~\eqref{eqn:3.3a} we have that if $\mathbf{X}\in\omega_{1}$
\begin{equation}
	\begin{split}
		\Var{U_{T}(\mathbf{X})} &= 
		\tau_{1} \int_{0}^{T} \left\{
		\log\left(\frac{\lambda_{1}(t)}{\lambda_{2}(t)}\right)
		\right\}^2 p_{1}(t)dt \\
		& \leq \tau_{1} \log^2 \left(\frac{C}{\delta}\right)
	\end{split}
	\label{eqn:3.11} \,.
\end{equation}
On the other hand by~\eqref{eqn:3.9} and Lemma~\ref{lemma:3}, we get
\begin{equation*}
	\begin{split}
		\Var{U_{T}(\mathbf{X})} & \geq
		\tau_{1} \bigg\{
		\mathbf{K}_{T}^{2}\left(p_{1} \parallel p_{2}\right) 
		\bigg. \\
		& \left.
		+ 2 \log\left(\frac{\tau_{1}}{\tau_{2}}\right)
		\mathbf{K}_{T}\left(p_{1} \parallel p_{2}\right) +
		\log^{2}\left(\frac{\tau_{1}}{\tau_{2}}\right)
		\right\}
	\end{split}
	\,.
\end{equation*}
The right-hand side of this inequality is equal to \mbox{$
	\tau_{1}\left(
	\int_{0}^{T}
	\log\left(\frac{\lambda_{1}(t)}{\lambda_{2}(t)}\right)
	p_{1}(t)dt
	\right)^{2}
	$} and by~\eqref{eqn:3.3a} this is not smaller than \mbox{$
	\tau_{1}\log^{2}\left(\frac{\delta}{C}\right)
	$}. Hence, if $\mathbf{X}\in\omega_{1}$ this gives the following bounds
\begin{equation}
	\tau_{1}\log^{2}\left(\frac{\delta}{C}\right) \leq
	\Var{U_{T}(\mathbf{X})} \leq
	\tau_{1}\log^{2}\left(\frac{C}{\delta}\right)
	\label{eqn:3.11a}\,.
\end{equation}
Analogously, we can show that if $\mathbf{X}\in\omega_{2}$, then
\begin{equation}
	\tau_{2}\log^{2}\left(\frac{\delta}{C}\right) \leq
	\Var{U_{T}(\mathbf{X})} \leq
	\tau_{2}\log^{2}\left(\frac{C}{\delta}\right)
	\label{eqn:3.12}\,.
\end{equation}

The bounds in~\eqref{eqn:3.11a}, \eqref{eqn:3.12} and the assumption
in~\eqref{eqn:3.10} lead to the following limit behavior of
$\Var{U_{T}(\mathbf{X})}$.
\begin{lemma}
	\label{lemma:4}
	Let the assumptions $\assumption{1}$, $\assumption{2}$ hold. Then for
	$\mathbf{X}\in\omega_{1}$ or $\mathbf{X}\in\omega_{2}$ we have
	\begin{equation}
		\begin{split}
			d\log^{2}\left(\frac{\delta}{C}\right) 
			& \leq
			\underline{\lim}_{T\to\infty}
			\Var{\frac{1}{\sqrt{T}}U_{T}\left(\mathbf{X}\right)} \\
			& \leq
			\overline{\lim}_{T\to\infty}
			\Var{\frac{1}{\sqrt{T}}U_{T}\left(\mathbf{X}\right)} \leq
			d\log^{2}\left(\frac{C}{\delta}\right)
		\end{split}
		\label{eqn:3.12b} \,.
	\end{equation}
\end{lemma}

The question whether the inferior and superior limits in~\eqref{eqn:3.12b} are equal
remains open. It should be noted that if~\eqref{eqn:3.10} is in the form \mbox{$
	\frac{1}{T}	\int_{0}^{T} \lambda_{i}(u)du \to d_{i}
	\enspace i=1,2
	$}, then the result of Lemma~\ref{lemma:4} holds with $d$ replaced by $d_{1}$
(if $\mathbf{X}\in\omega_{1}$) or $d_{2}$ (if $\mathbf{X}\in\omega_{2}$), respectively.
To shed some light on the result in \eqref{eqn:3.12b} let us consider the following
simple example.
\begin{example} 
	\label{example:1}
	Let us consider the classification problem with the intensity functions 
	$\lambda_1(t)$ and $\lambda_2(t) = \mu \lambda_1(t)$ for some $\mu > 0$. 
	Then we have $\tau_2 = \mu\tau_1$ and $p_2(t) = p_1(t)$. This implies that the 
	condition in \eqref{eqn:3.10} reads as \mbox{$
		\frac{1}{T}	\int_{0}^{T} \lambda_{1}(u)du \to d
		$}	and \mbox{$\frac{1}{T} \int_{0}^{T} \lambda_{2}(u)du \to \mu d$}. Then, a simple 
	algebra gives the following analog of Lemma~\ref{lemma:4}.\\
	If $\mathbf{X}\in\omega_{1}$ then
	\begin{equation}
		\lim_{T\to\infty} \Var{\frac{1}{\sqrt{T}}U_{T}\left(\mathbf{X}\right)} = 
		d \log^{2}(\mu)  
		\label{eqn:3.12c}\,,
	\end{equation}
	whereas if $\mathbf{X}\in\omega_{2}$ then
	\begin{equation}
		\lim_{T\to\infty} \Var{\frac{1}{\sqrt{T}}U_{T}\left(\mathbf{X}\right)} = 
		d \mu \log^{2}(\mu)
		\label{eqn:3.12d}\,.
	\end{equation}
	Note that the assumption $\assumption{1}$ is not required here. Also if 
	$\mu=1$ then the asymptotic constants are zero, i.e., this corresponds to the case
	$\lambda_1(t) = \lambda_2(t)$. Moreover, the asymptotic constants tend to infinity
	as $\mu \to \infty$.
\end{example}

An important consequence of Lemma~\ref{lemma:4} is the following weak law of large 
numbers for the average value of $U_{T}(\mathbf{X})$ defined in~\eqref{eqn:3.9a}.
\begin{theorem}
	\label{theorem:1}
	Let the conditions of Lemma~\ref{lemma:4} hold. Then for $\mathbf{X}$ coming 
	either from class $\omega_{1}$ or class $\omega_{2}$ we have
	\begin{equation}
		\frac{1}{T}U_{T}(\mathbf{X}) =
		\frac{1}{T} \int_{0}^{T} \log\left(
		\frac{\lambda_{1}(t)}{\lambda_{2}(t)}
		\right) dM(t) \to 0
		\enspace\enspace
		(P)
		\label{eqn:3.13}
	\end{equation}
	as $T\to\infty$.
\end{theorem}
The proof of this fact is a direct application of Lemma~\ref{lemma:4} and the 
Chebyshev inequality. In fact, let us consider the case $\mathbf{X}\in\omega_{1}$. 
Then, for any $\epsilon>0$ we have
\begin{equation}
	\mathbf{P}\left(\frac{1}{T}\left|U_{T}(\mathbf{X})\right|\geq\epsilon\right) \leq
	\frac{\Var{U_{T}(\mathbf{X})}}{T^2\epsilon^2}
	\label{eqn:3.14}\,.
\end{equation}
The right-hand side of~\eqref{eqn:3.14} is equal to 
$\Var{\frac{1}{\sqrt{T}}U_{T}(\mathbf{X})}/T\epsilon^2$, where due to~\eqref{eqn:3.12b} 
the limit superior of $\Var{\frac{1}{\sqrt{T}}U_{T}(\mathbf{X})}$ is bounded by a finite
constant. This confirms the claim of Theorem~\ref{theorem:1}.

Our next goal is to strengthen the result of Theorem~\ref{theorem:1} by establishing 
the strong law of large numbers. This will result directly from the exponential 
inequality for the average of $U_{T}(\mathbf{X})$ defined in~\eqref{eqn:3.9a}. Our 
main tools here are exponential inequalities for martingales of counting processes 
established recently in~\cite{le2021exponential}, see also~\cite{van1995exponential} 
for earlier results. Hence, we employ the following adapted to our needs version 
of Theorem~5 in~\cite{le2021exponential}, see Appendix~B  for details.
\begin{lemma}
	\label{lemma:5}
	Let $N(t)$ be the counting process allowing the decomposition in~\eqref{eqn:2.9}. 
	Let $U_{T}=\int_{0}^{T}g(t)dM(t)$ be the stochastic integral of the real-valued 
	function $g(t)$ with respect to the martingale $M(t)$ increments. Suppose that
	\begin{enumerate}[(a)]
		\item $\left|g(t)\right| \leq u_{T}$ for all $t\in[0,T]$.
		\item $\int_{0}^{T}g^2(t)\lambda(t)dt \leq v_T$,
	\end{enumerate}
	where $u_T$ and $v_T$ are some finite constants. Then, for each $\epsilon>0$ we 
	have
	\begin{equation}
		\mathbf{P}\left(\left|U_{T}\right|\geq\epsilon\right) \leq
		2\exp\left[
		-\frac{\epsilon^2}{2v_{T}+u_{T}\epsilon}
		\right]
		\label{eqn:3.15}\,.
	\end{equation}
	It is worth noting that this bound holds for any finite $T$.
\end{lemma}

Lemma~\ref{lemma:5} can be directly applied for the evaluation of the stochastic 
integral in~\eqref{eqn:3.9a}. In fact, with \mbox{$
	g(t) = \log\left(\frac{\lambda_{1}(t)}{\lambda_{2}(t)}\right)
	$} and by the assumption $\assumption{1}$, we have that \mbox{$
	|g(t)| \leq \log\left(\frac{C}{\delta}\right)
	$}. Hence, the condition (a) in Lemma~\ref{lemma:5} is met with 
$u_{T}=\log\left(\frac{C}{\delta}\right)$ for all $T>0$. By virtue of the 
property~\eqref{eqn:a7} in  Appendix~A the integral in the condition (b) of
Lemma~\ref{lemma:5} reads as
\begin{equation}
	\Var{U_{T}(\mathbf{X})} =
	T \Var{\frac{1}{\sqrt{T}}U_{T}(\mathbf{X})} =
	T\theta_{T}
	\label{eqn:3.16}\,,
\end{equation}
where due to~\eqref{eqn:3.12b} the limit superior of $\theta_T$ is bounded by a finite
constant.

The preceding discussion gives the following exponential bound for the average value of
$U_{T}(\mathbf{X})$ defined in~\eqref{eqn:3.9a}. The bound is valid for any finite 
$T > 0$.
\begin{lemma}
	\label{lemma:6}
	Suppose that the assumption $\assumption{1}$ holds. Then for $\mathbf{X}$ coming 
	either from class $\omega_{1}$ or class $\omega_{2}$ and every $\epsilon>0$ we have
	\begin{equation}
		\mathbf{P}\left(
		\frac{1}{T}\left|U_{T}(\mathbf{X})\right| \geq \epsilon
		\right) \leq
		2 \exp\left[
		-T\frac{\epsilon^2}{2\theta_{T}+u\epsilon}
		\right]
		\label{eqn:3.17}\,,
	\end{equation}
	where $u=\log\left(\frac{C}{\delta}\right)$ and the factor $\theta_{T}$ is defined 
	in~\eqref{eqn:3.16}.
\end{lemma}
The exponential bound  in~\eqref{eqn:3.17} and the Borel-Cantelli lemma yield the
following strong version of Theorem~\ref{theorem:1}. We should note, however,
that the Borel-Cantelli lemma applies to a sequence of random variables, while the
random variable $\xi_T = U_{T}(\mathbf{X})/T$ is a function of the continuous
parameter~$T$. Nevertheless, one can discretize~$\xi_T$ by finding a sequence of
times~$T_n$, such that $T_n \to \infty$ as $n \to \infty$ and then employ the standard
Borel-Cantelli lemma. We refer to~\cite{vere1982estimation} for details for such
discretization strategy.
\begin{theorem}
	\label{theorem:2}
	Let the assumptions $\assumption{1}$ and $\assumption{2}$ hold. 
	Then for $\mathbf{X}$ coming either from 
	class $\omega_{1}$ or class $\omega_{2}$ we have
	\begin{equation}
		\frac{1}{T}U_{T}(\mathbf{X}) =
		\frac{1}{T} \int_{0}^{T} \log\left(
		\frac{\lambda_{1}(t)}{\lambda_{2}(t)}
		\right) dM(t) \to 0
		\enspace\enspace
		(a.s.)
		\label{eqn:3.18}
	\end{equation}
	as $T\to\infty$.
\end{theorem}


\subsection{The Bayes Risk}
\label{subsec:03_bayes_bounds/risk}

\noindent In this section we wish to evaluate the Bayes risk defined in~\eqref{eqn:2.8}. 
Our analysis will employ the results obtained 
in Section~\ref{subsec:03_bayes_bounds/decision}. Owing 
to~\eqref{eqn:2.8} it suffices to consider the probability of misclassification
$\mathbf{P}\left(W_{T}(\mathbf{X})\geq\eta_{T}|\mathbf{X}\in\omega_{2}\right)$. The 
analysis of the probability 
$\mathbf{P}\left(W_{T}(\mathbf{X})<\eta_{T}|\mathbf{X}\in\omega_{1}\right)$ is
analogous. By virtue of~\eqref{eqn:3.1} we can write 
\begin{equation}
	\begin{split}
		& \mathbf{P}\left(W_{T}(\mathbf{X})\geq\eta_{T}|\mathbf{X}\in\omega_{2}\right) \\
		& = 
		\mathbf{P}\left(U_{T}(\mathbf{X})\geq\alpha_{T} +
		\log\left(\frac{\pi_{2}}{\pi_{1}}\right)
		|\mathbf{X}\in\omega_{2}\right)
	\end{split}
	\label{eqn:3.19}\,,
\end{equation}
where $U_{T}(\mathbf{X})$ is defined in~\eqref{eqn:3.2} and $\alpha_{T}$
(under the fact that $\mathbf{X}\in\omega_{2}$) is given by
\begin{equation}
	\alpha_{T} =
	\tau_{1} - \tau_{2} + \tau_{2}
	\log\left(\frac{\tau_{2}}{\tau_{1}}\right) +
	\tau_{2}\mathbf{K}_T(p_{2} \parallel p_{1})
	\label{eqn:3.20}\,.
\end{equation}
The first result reveals that the Bayes risk tends to zero as $T\to\infty$ under
the assumptions $\assumption{1}$ and $\assumption{2}$. 
This is the direct consequence of the weak law of large numbers 
established in Theorem~\ref{theorem:1}, see~\eqref{eqn:3.13}. \\
Hence, we have the 
following convergence result that also gives the upper bound for the Bayes risk.
\begin{theorem}
	\label{theorem:3}
	Let the assumptions $\assumption{1}$ and $\assumption{2}$ hold. Then, we have
	\begin{equation*}
		\mathbf{R}_{T}^{*} \to 0 \enspace \text{as} \enspace T\to\infty
		\,.
	\end{equation*}
	Furthermore,
	\begin{equation}
		\mathbf{R}_{T}^{*} \leq 
		\left(\pi_{1}a_{T} + \pi_{2}b_{T}\right) \frac{1}{T}
		\label{eqn:3.21} \,,
	\end{equation}
	for some finite constants $a_{T}$, $b_{T}$.
\end{theorem}
The proof of Theorem~\ref{theorem:3} is deferred to Appendix~B, where also the
explicit expressions for $a_{T}$ and $b_{T}$ are given. 
The bound in~\eqref{eqn:3.21} is obtained by utilizing only the second
moment of the stochastic integral $U_{T}(\mathbf{X})$
in~\eqref{eqn:3.2}.
\begin{remark}
	\label{remark:1}
	Hence under the assumptions $\assumption{1}$ and $\assumption{2}$ the Bayes risk tends
	to zero with the rate $1/T$. The proof of Theorem~\ref{theorem:3} reveals also the
	following form of the asymptotic constant
	\begin{equation}
		c_{1} = \frac{1}{d} \left(
		\frac{\log(C/\delta)}{\log(\delta/C)}
		\right)^{2}
		\label{eqn:3.22} \,.
	\end{equation}
	Hence, for large $T$ one can write $\mathbf{R}_{T}^{*} \prec c_{1}\frac{1}{T}$.
\end{remark}

By virtue of the result of Lemma~\ref{lemma:6} we can substantially improve the bound
in~\eqref{eqn:3.21}. Hence, we have the following result.
\begin{theorem}
	\label{theorem:4}
	Let the assumptions $\assumption{1}$ and $\assumption{2}$ hold. Then, we have
	\begin{equation}
		\mathbf{R}_{T}^{*}  \leq
		\pi_{1}\exp\left[-A_{T}T\right] + \pi_{2}\exp\left[-B_{T}T\right]
		\label{eqn:3.23} \,,
	\end{equation}
	for some finite constants $A_{T}$, $B_{T}$.
\end{theorem}

The proof of Theorem~\ref{theorem:4} is deferred to Appendix~B, where also the
explicit expressions for $A_{T}$ and $B_{T}$ are presented.

\begin{remark}
	\label{remark:2}
	The proof of Theorem~\ref{theorem:4} shows that using the exponential inequality
	for the martingale process the Bayes risk tends to zero with the exponential
	rate and the	following asymptotic constant
	\begin{equation}
		c_{2} = d\frac{1}{3}\left(
		\frac{\log(\delta/C)}{\log(C/\delta)}
		\right)^{2}
		\label{eqn:3.24} \,,			
	\end{equation}
	where $(\delta,C)$ characterizes the assumption $\assumption{1}$, whereas $d$ appears
	in the assumption $\assumption{1}$. Hence, for large $T$ one can write \mbox{$
		\mathbf{R}_{T}^{*} \prec \exp\left[-c_{2}T\right]
		$}.
	It is also worth noting that larger $d$ in the assumption $\assumption{2}$ makes the
	bounds in~\eqref{eqn:3.21} and~\eqref{eqn:3.23} tighter. In fact, the constant~$c_{1}$
	in~\eqref{eqn:3.22} decreases with $d$, whereas the constant
	$c_{2}$ in~\eqref{eqn:3.24} increases with $d$.
\end{remark}

\begin{example} 
	\label{example:2}
	Consider the classification problem discussed in Example~\ref{example:1}. 
	Then, using the results in~\eqref{eqn:3.12c} and~\eqref{eqn:3.12d} and some algebra
	we can show the following counterpart of the result of Theorem~\ref{theorem:4}
	\begin{equation}
		\mathbf{R}_{T}^{*} \prec 
		\pi_1 \exp\left[-c_{1}(\mu)dT\right] + 
		\pi_2 \exp\left[-c_{2}(\mu)dT\right]
		\label{eqn:3.25a}\,.
	\end{equation}
	The asymptotic constants $c_1 (\mu )$, $c_2 (\mu)$ can be written in the explicit
	form and they obey the following properties
	\begin{equation*}
		\lim_{\mu \to 1} c_1 (\mu ) = \lim_{\mu \to 1} c_2 (\mu ) = 0
	\end{equation*}
	and
	\begin{equation*}
		\lim_{\mu \to \infty} c_1 (\mu ) = \lim_{\mu \to \infty} c_2 (\mu ) = \infty. 
	\end{equation*}
	The former limit corresponds to the indistinguishable case, i.e., 
	$\lambda_1(t) = \lambda_2(t)$. On the other hand, the latter limit exhibits that
	if $\mu \to \infty$ then $\mathbf{R}_{T}^{*} \to 0$. Again the assumption
	$\assumption{1}$ is not needed here.
\end{example}

\begin{remark}
	\label{remark:3}
	In~\cite{rong2021error} the following upper bound for the Bayes risk is given
	\begin{equation*}
		\mathbf{R}_{T}^{*} \leq
		\sqrt{\pi_{1}\pi_{2}} \exp\left(-\beta(T)\right)
		\,,
	\end{equation*}
	where \mbox{$
		\beta(T) = \int_{0}^{T} \left[
		\frac{1}{2} \lambda_{1}(u) + \frac{1}{2} \lambda_{2}(u) -
		\sqrt{\lambda_{1}(u)\lambda_{2}(u)}
		\right]	du
		$} is a positive factor. This is the
	classical Bhattacharya bound~\cite{devroye2013probabilistic} extended to the
	classification problem for point processes. The behavior of $\beta(T)$ under the
	condition $\assumption{2}$ is an interesting open question.
	In the special case examined in Examples~\ref{example:1},~\ref{example:2}
	we can show that $\beta(T)$ behaves asymptotically as $c(\mu )dT$, where \mbox{$
		c(\mu ) = (\mu +1)/2 - \sqrt{\mu}
		$}. Interestingly $c(\mu ) \geq  c_1 (\mu ), c_2 (\mu)$, where 
	$c_1 (\mu), c_2 (\mu )$ appear in our bound in~\eqref{eqn:3.25a}.
\end{remark}

\begin{remark}
	\label{remark:4}
	The convergence of the Bayes risk $\mathbf{R}_{T}^{*}$ to zero is determined by 
	the condition in $\assumption{2}$. This is due to the fact that the class 
	intensity functions $\lambda_{1}(t),\lambda_{2}(t)$ grow with increasing $T$. If 
	$\assumption{2}$ does not hold, e.g, if $\lambda_{1}(t),\lambda_{2}(t)$ are 
	compactly supported then the convergence of $\mathbf{R}_{T}^{*}$ to zero is 
	impossible. In this case in order to enforce the grow of 
	$\lambda_{1}(t),\lambda_{2}(t)$ one could use the multiplicative model due to 
	Aalen~\cite{aalen1978nonparametric}, i.e., we consider
	\begin{equation}
		\lambda_{i}(t)=d\gamma_{i}(t),
		\enspace
		i=1,2
		\label{eqn:3.27}\,,		
	\end{equation}
	where $\gamma_{i}(t)$\enspace$i=1,2$ are fixed functions and $d$ is a parameter 
	that is allowed to grow. It is an interesting alternative 
	to derive the results obtained in this paper under the multiplicative class 
	intensity model in~\eqref{eqn:3.27}.
\end{remark}

In the following example we give some numerical illustration of the aforementioned
results.
\begin{example}
	\label{example:3}
	Let us model the class intensities in the following form
	\begin{equation}
		\begin{split}
			\lambda(t;\phi) = {}
			& 1.6+\cos\left(\frac{\pi}{4\sqrt{3}}t+\phi\right) \\
			&+ 0.5\cos\left(\frac{\pi}{3\sqrt{2}}t+\frac{\pi}{4} + \phi\right)
		\end{split}
		\label{eqn:sim1} \,.
	\end{equation}
	Various choices of~$\phi$ define~$\lambda_{1}(t)$,~$\lambda_{2}(t)$. 
	\figurename~\ref{fig:03_bayes_bounds/risk/intensity} depicts 
	\mbox{$\lambda_{1}(t)=\lambda\left(t; \frac{\pi}{16}\right)$} and 
	\mbox{$\lambda_{2}(t)=\lambda\left(t; \frac{\pi}{2}\right)$}.
	
	\begin{figure}[!t]
		\includegraphics[width=\figurewidth]{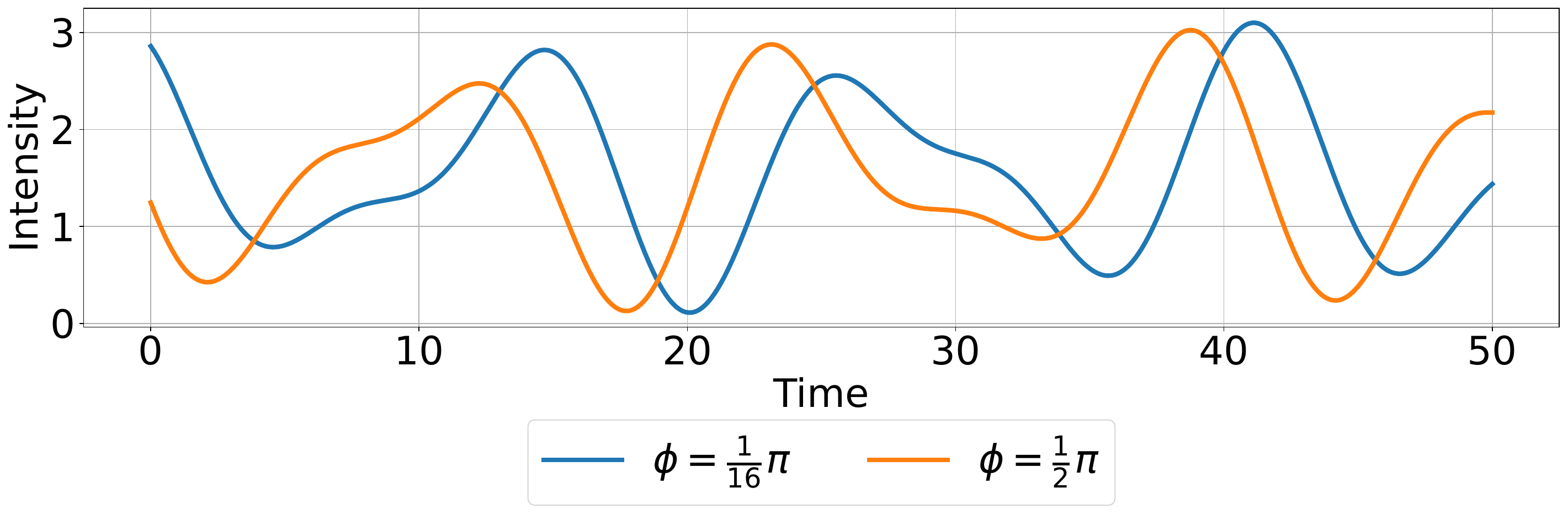}
		\centering
		\caption{Intensity functions \mbox{
				$\lambda_{1}(t)=\lambda\left(t; \frac{\pi}{16}\right)$
			} and \mbox{
				$\lambda_{2}(t)=\lambda\left(t; \frac{\pi}{2}\right)$
			}.
		}
		\label{fig:03_bayes_bounds/risk/intensity}
	\end{figure}
	
	Figure~\ref{fig:03_bayes_bounds/risk/simulations} illustrates the fact that 
	the Bayes risk tends to zero as~$T$ gets larger. The model of class intensities  
	defined in~\eqref{eqn:sim1} is parametrized by $\phi$, i.e., we set 
	\mbox{$\lambda_{1}(t)=\lambda_{1}(t;\phi_{1})$} and 
	\mbox{$\lambda_{2}(t)=\lambda_{1}(t;\phi_{2})$}. The slowest decay of 
	$\mathbf{R}_T^{*}$ is seen for very close intensities, i.e., when
	\mbox{$\phi_{2}/\phi_{1}=2$}
	(in red), whereas the fast rate of convergence is observed for distant intensities,
	i.e., when
	\mbox{$\phi_{2}/\phi_{1}=16$} (in blue). Nevertheless, since $\lambda(t;\phi)$ 
	in~\eqref{eqn:sim1} meets the assumptions $\assumption{1}$ and $\assumption{2}$ 
	we can observe the exponential rate of convergence.
	\begin{figure}[!t]
		\includegraphics[width=\figurewidth]{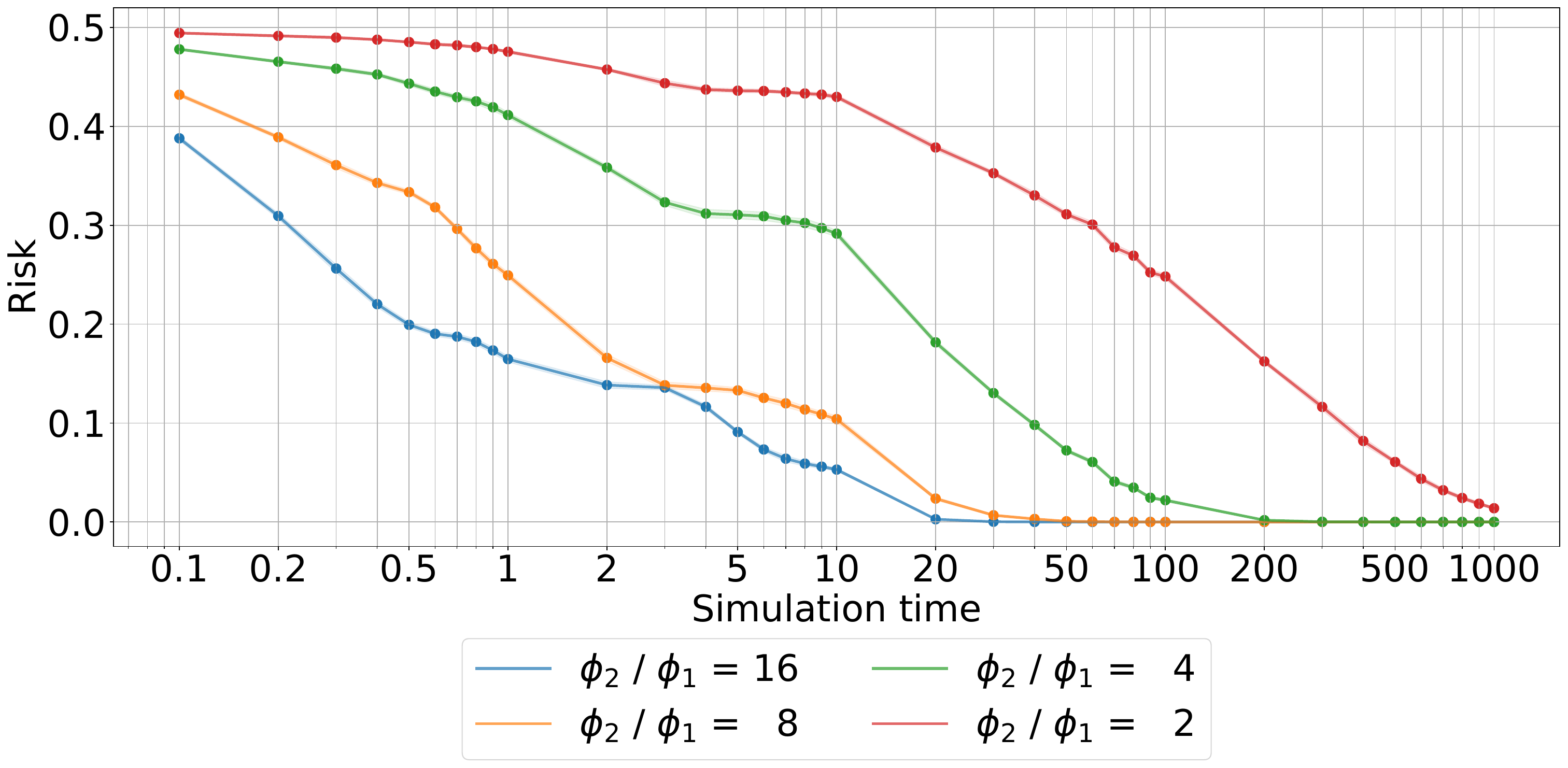}
		\centering
		\caption{The Bayes risk~$\mathbf{R}^{*}_{T}$ versus~$T$ for a two-class 
			classification problem.}
		\label{fig:03_bayes_bounds/risk/simulations}
	\end{figure}
\end{example}

\section{Nonparametric Classification Rules}
\label{sec:04_nonparametric}

\subsection{Plug-in Classifiers}
\label{subsec:04_nonparametric/plug-in}

\noindent In practice one does not know the true class intensities functions and must rely 
on some training data in order to form a data-driven classification rule. In this paper we
apply the plug-in strategy to design a classifier, i.e. the classifier that is the 
empirical counterpart of the optimal Bayes rule in~\eqref{eqn:2.4} 
or equivalently in~\eqref{eqn:2.6}.
We have already pointed out that the single-sample 
based intensity function estimate cannot be consistent unless there is a certain 
mechanism that makes the intensity function increase, e.g., the multiplicative model
in~\eqref{eqn:3.27}. In this paper we consider the intensity model based on the
increasing number of replicates of the class spiking processes. Hence, contrary to the
results of Section~\ref{sec:03_bayes_bounds} the observation interval $[0,T]$ is kept
constant.

Hence, let \mbox{
	$\mathbf{D}_{L}=\{(\mathbf{X}_{1},Y_{1}),\ldots,(\mathbf{X}_{L},Y_{L})\}$
} be the learning sequence being a sample of~$L$ independent observations of the 
labeled spiking processes $(\mathbf{X},Y)$. 
Here~$\mathbf{X}_{j}$ is the variable-length vector, i.e., \mbox{
	$\mathbf{X}_{j}=\left[t_{1}^{[j]},\ldots,t_{N^{[j]}}^{[j]};N^{[j]}\right]$
} and \mbox{$Y_{j}\in \{\omega_{1},\omega_{2}\}$}, where \mbox{$N^{[j]}=N^{[j]}(T)$}. 
Hence, all data are measured in the fixed time window~$[0,T]$. Let \mbox{
	$L_{1}$, $L_{2}$
} be the number of training data of classes~$\omega_{1}$ and~$\omega_{2}$, 
respectively. 

We wish to form the plug-in classification rule based on the optimal 
decision given in~\eqref{eqn:2.6}. This requires estimating the class intensity 
functions \mbox{$\lambda_{1}(t)$, $\lambda_{2}(t)$}, or equivalently the shape 
densities \mbox{$p_{1}(t)$, $p_{2}(t)$} and the corresponding intensity factors 
\mbox{$\tau_{1}$, $\tau_{2}$}. It is known that the prior probabilities can be 
estimated by \mbox{$\widehat{\pi}_{1}=L_{1}/L$} and 
\mbox{$\widehat{\pi}_{2}=L_{2}/L$}. In order to estimate \mbox{$
	\left\{\left(\tau_{i},p_{i}(t)\right), i=1,2\right\}
	$} one can begin with the use of the single 
sample~$\mathbf{X}_{j}$. Note that 
\mbox{$\E{N^{[j]}|Y_j=\omega_{i}}=\tau_{i}$} and one can form the 
unbiased estimate of~$\tau_{i}$ as \mbox{$\widehat{\tau}_{i}^{[j]}=N^{[j]}$}. However, 
\mbox{$\Var{N^{[j]}|Y_j=\omega_{i}}=\tau_{i}$} and this is an 
inconsistent estimate of~$\tau_{i}$. The latter fact results from the local Poisson 
behavior of the spiking process, see Appendix~A. Nevertheless, the aggregation
of~$\{\widehat{\tau}_{i}^{[j]}\}$ leads to consistent estimate of~$\tau_{i}$ for the 
increased size of the training set. Hence, let 
\begin{equation}
	\widehat{\tau}_{i}=
	\frac{1}{L_{i}}\sum_{j=1}^{L}N^{[j]}\mathbf{1}(Y_{j}=\omega_{i})
	\label{eqn:4.1}
\end{equation}
be an estimate of~$\tau_{i}$, \enspace$i=1,2$. 
In the analogous way we can deal with the problem of estimating~$p_{i}(t)$. Let 
\mbox{$\widehat{p}_{i}^{[j]}(t)$} be a certain nonparametric estimate 
of~$p_{i}(t)$ based on the single sample~$\mathbf{X}_{j}$ from the class $\omega_{i}$. 
Then, the aggregated estimate of~$p_{i}(t)$ takes the following form 
\begin{equation}
	\widehat{p}_{i}(t)=
	\frac{1}{L_{i}}\sum_{j=1}^{L}\widehat{p}_{i}^{[j]}(t)\mathbf{1}(Y_{j}=
	\omega_{i}) ,
	\enspace i=1,2
	\label{eqn:4.2} \,.
\end{equation}
Plugging~\eqref{eqn:4.1} and~\eqref{eqn:4.2} into~\eqref{eqn:2.6} gives us the 
following empirical classification rule \mbox{$\widehat{\psi}_{L,T}$}: 
classify~$\mathbf{X}=\left[t_{1},\ldots,t_{N};N\right]\in\omega_{1}$ if 
\begin{equation}
	\widehat{W}_{L,T}(\mathbf{X})\geq\widehat{\eta}_{L,T}
	\label{eqn:4.3} \,,
\end{equation}
where 
$
\widehat{W}_{L,T}(\mathbf{X})=
\sum_{i=1}^{N}\log\left(
\frac{\widehat{p}_{1}(t_{i})}{\widehat{p}_{2}(t_{i})}
\right)
$,
$
\widehat{\eta}_{L,T}=
\widehat{\tau}_{1}-\widehat{\tau}_{2}+
N\log\left(\frac{\widehat{\tau}_{2}}{\widehat{\tau}_{1}}\right)+
\log\left(\frac{L_{2}}{L_{1}}\right)
$.
In Section~\ref{subsec:04_nonparametric/kernel} we propose a concrete kernel-type 
estimate of the shape densities.

In this section we present a general result on the convergence of the 
rule~$\widehat{\psi}_{L,T}$ to the Bayes decision~$\psi_{T}^{\star}$. This result is 
in the spirit of the Bayes risk consistency theorem established 
in~\cite{greblicki1978asymptotically} in the context of the standard fixed dimension 
data sets. Let us first consider the pointwise behavior of the 
rule~$\widehat{\psi}_{L,T}$ in~\eqref{eqn:4.3}. Hence, let \mbox{$
	\mathbf{P}\left(\widehat{\psi}_{L,T}(\mathbf{x}) = \psi_{T}^{\star}(\mathbf{x})\right)
	$} be the probability that the empirical rule makes the same decisions as the optimal 
Bayes rule for a fixed test vector $\mathbf{x}$. Our first result reveals that this 
probability tends to one if the size of the training set tends to infinity.
\begin{theorem}
	\label{theorem:5}
	Suppose that  for $i = 1,2$ and $L\to\infty$ the following property holds
	\begin{equation}
		\widehat{p}_{i}(t) \to p_{i}(t)
		\enspace 
		(P)
		\enspace 
		\text{uniformly on}
		\enspace
		[0, T]
		\label{eqn:4.4}\,.
	\end{equation}
	Then,
	\begin{equation*}
		\mathbf{P}\left(\widehat{\psi}_{L,T}(\mathbf{x}) =
		\psi_{T}^{\star}(\mathbf{x})\right) \to 1
	\end{equation*}
	as $L\to\infty$. The proof of Theorem~\ref{theorem:5} is given in Appendix~C.
	This result assures that~$\widehat{\psi}_{L,T}$ converges to $\psi_{T}^{\star}$ as 
	long as one can construct uniformly consistent estimates 
	of~$p_{i}(t)$, \enspace$i = 1,2$. Clearly, the uniform convergence of estimates of 
	the class intensity functions $\lambda_{i}(t)$ also	implies the local consistency 
	result of Theorem~\ref{theorem:5}.
\end{theorem} 
The proof of Theorem~\ref{theorem:5} 
reveals also that the 0-1 distance between~$\widehat{\psi}_{L,T}(\mathbf{x})$ 
and~$\psi_{T}^{\star}(\mathbf{x})$ tends to zero. Hence, we have
\begin{equation}
	\rho(
	\widehat{\psi}_{L,T}(\mathbf{x}), \psi_{T}^{\star}(\mathbf{x})
	) \to 0
	\enspace (P)
	\label{eqn:4.5}
\end{equation}
as $L\to\infty$, where 
\begin{equation*}
	\rho(
	\widehat{\psi}_{L,T}(\mathbf{x}), \psi_{T}^{\star}(\mathbf{x})
	) = \indicator{
		\widehat{\psi}_{L,T}(\mathbf{x}) \neq \psi_{T}^{\star}(\mathbf{x})
	}
	\,.
\end{equation*}

The condition in~\eqref{eqn:4.4} of Theorem~\ref{theorem:5} assures that the decision 
function~$\widehat{W}_{L,T}(\mathbf{x})$ in~\eqref{eqn:4.3} tends to the optimal 
decision function~$W_{T}(\mathbf{x})$ in~\eqref{eqn:2.6}. This is the convergence
needed in the proof of Theorem~\ref{theorem:5} and is summarized in the following lemma.
\begin{lemma}
	\label{lemma:7}
	Let the class intensities $\lambda_{1}(t),\lambda_{2}(t)$ be uniformly continuous 
	on~$[0,\infty)$ such that restricted to $[0,T]$ satisfy 
	the assumption~$\assumption{1}$. Let~\eqref{eqn:4.4} hold. Then, we have
	\begin{equation}
		\widehat{W}_{L,T}(\mathbf{x}) \to W_{T}(\mathbf{x})
		\enspace (P)
		\label{eqn:4.6}\,,
	\end{equation}
	as $L\to\infty$. 
\end{lemma}
The proof of Lemma~\ref{lemma:7} is postponed to Appendix~C.
The convergence in~\eqref{eqn:4.6} is uniform with respect to~$\mathbf{x}$.

The classification rule~$\widehat{\psi}_{L,T}$ in~\eqref{eqn:4.3} is also characterized 
by the threshold value~$\widehat{\eta}_{L,T}$. Note that~$\pi_{1}=L_{1}/L$, 
\enspace$\pi_{2}=L_{2}/L$ are weakly consistent estimates of the prior probabilities
$\pi_{1}, \pi_{2}$. Also the aggregated estimate~$\tau_{i}$ in~\eqref{eqn:4.1} of the 
intensity factor~$\tau_{i}$ is weakly consistent. Hence, the preceding discussion gives 
the following consistency result
\begin{equation}
	\widehat{\eta}_{L,T} \to \eta_{T}
	\enspace (P)
	\label{eqn:4.7}
\end{equation}
as $L\to\infty$.\\
The local consistency of~$\widehat{\psi}_{L,T}$ leads to the global convergence 
characterized by the conditional risk. Hence, let \mbox{$
	\mathbf{P}(
	\widehat{\psi}_{L,T}(\mathbf{X}) \neq Y | \mathbf{D}_{L}
	) = \E{\indicator{
			\widehat{\psi}_{L,T}(\mathbf{X}) \neq Y
		} | \mathbf{D}_{L}}
	$} be the conditional risk associated with the rule~$\widehat{\psi}_{L,T}$.
Since \mbox{$
	\mathbf{R}_{T}^{\star} = 
	\E{\indicator{
			\psi_{T}^{\star}(\mathbf{X}) \neq Y
	}}
	$}, then one can write
\begin{equation*}
	\begin{split}
		0 & \leq \mathbf{R}_{L,T} - \mathbf{R}_{T}^{\star} \\
		& = 
		\E{
			\indicator{\widehat{\psi}_{L,T}(\mathbf{X}) \neq Y} -
			\indicator{\psi_{T}^{\star}(\mathbf{X}) \neq Y}
			| \mathbf{D}_{L}}
	\end{split}
	\,.
\end{equation*}
Recalling the definition of the distance in~\eqref{eqn:4.5} the above is bounded by
\begin{equation*}
	\E{\rho(
		\widehat{\psi}_{L,T}(\mathbf{X}), \psi_{T}^{\star}(\mathbf{X})
		) | \mathbf{D}_{L}
	}
	\,.
\end{equation*}
Owing to~\eqref{eqn:4.5} and Lebesgue's dominated convergence theorem we obtain the 
main result of this section.
\begin{theorem}
	\label{theorem:6}
	Consider the class of plug-in classifiers defined in~\eqref{eqn:4.3}. Suppose that 
	the	conditions of Theorem~\ref{theorem:5} hold. Then, we have the following Bayes
	risk consistency result
	\begin{equation}
		\mathbf{R}_{L,T} \to \mathbf{R}_{T}^{\star}
		\enspace (P)
		\label{eqn:4.8}
	\end{equation}
	as $L\to\infty$.
\end{theorem}


\subsection{Kernel Classifiers}
\label{subsec:04_nonparametric/kernel}

\noindent It is known~\cite{andersen2012statistical, diggle1988equivalence} that the 
intensity function of a point process can be efficiently estimated by a class of kernel 
methods~\cite{wand1994kernel, grebpaw08}. In particular, the standard single sample  
kernel estimate of $\lambda_{i}(t)$ takes the form
\begin{equation}
	\widehat{\lambda}_{i}^{[j]}(t)=
	\sum_{l=1}^{N^{[j]}}K_{h}\left(t-t_{l}^{[j]}\right)
	\label{eqn:4.9}\,,
\end{equation}
where the sample \mbox{
	$\mathbf{X}_{j}=\left[t_{1}^{[j]},\ldots,t_{N^{[j]}}^{[j]};N^{[j]}\right]$
} comes from the class $\omega_{i}$. 

Here \mbox{$
	K_{h}(t) = h^{-1}K(t/h)
	$}, where the kernel $K(t)$ is assumed to be a compactly supported on $[-1,1]$, symmetric 
probability density function. For instance, one can choose the so-called
Epanechnikov kernel
\begin{equation*}
	K(t) = \frac{3}{4}\left(1-t^{2}\right)\indicator{|t|\leq1}
	\,.
\end{equation*}
The crucial tuning parameter $h$ is called the bandwidth as it controls the level of
smoothing via the scaled kernel $K_{h}(t)$.

The parameter $\tau_{i}$ can be estimated (from a single sample) by \mbox{$
	\widehat{\tau}_{i}^{[j]}=N^{[j]}
	$}. Therefore~\eqref{eqn:4.9} yields the following estimate of the shape density
\begin{equation*}
	\widehat{p}_{i}^{[j]}(t) = 
	\frac{1}{N^{[j]}} \sum_{l=1}^{N^{[j]}} K_{h}\left(t-t_{l}^{[j]}\right)
	\,.
\end{equation*}
As we have already pointed in Section~\ref{subsec:04_nonparametric/plug-in} the 
estimates $\widehat{\lambda}_{i}^{[j]}(t)$,\enspace $\widehat{p}_{i}^{[j]}(t)$ cannot be
consistent by merely increasing $T$. To overcome this problem one can utilize the
observed multiple training vectors and aggregate the single-sample estimates 
$\left\{\widehat{\lambda}_{i}^{[j]}\right\}$,\enspace 
$\left\{\widehat{p}_{i}^{[j]}\right\}$. This leads to the following aggregated kernel 
estimate of $p_{i}(t)$
\begin{equation}
	\widehat{p}_{i}(t) = 
	\frac{1}{L_i} 
	\sum_{j=1}^{L} \widehat{p}_{i}^{[j]}(t) \indicator{Y_{j}=\omega_{i}}
	\label{eqn:4.10}\,.
\end{equation}
Moreover, the aggregated estimate $\widehat{\tau}_{i}$ of $\tau_{i}$ is defined 
in~\eqref{eqn:4.1}. Plugging $\widehat{p}_{i}(t)$ and $\widehat{\tau}_{i}$, 
\enspace$i=1,2$ into~\eqref{eqn:4.3} we obtain the kernel classification rule. The 
aggregated kernel estimate $\widehat{\lambda}_{i}(t)$ of $\lambda_{i}(t)$ is defined in the 
analogous way, see~\eqref{eqn:4.13}.

Theorem~\ref{theorem:5} and Theorem~\ref{theorem:6} reveal that the sufficient
condition for the Bayes risk consistency is the convergence property 
in~\eqref{eqn:4.4}. Note that the statistical behavior of
$\widehat{p}_{i}(t)$ and $\widehat{\lambda}_{i}(t)$ is the same and therefore we can 
verify the requirement in~\eqref{eqn:4.4} for the
kernel intensity estimate. Hence, with some abuse of the notation let \mbox{$
	\left\{ \mathbf{X}_{1},	\mathbf{X}_{2}, \ldots, \mathbf{X}_{L} \right\}
	$} be the data set from the fixed class ($\omega_{1}$ or $\omega_{2}$)
of the counting process $N(t)$ characterized by 
the class intensity function $\lambda(t)$. Thus, one observes the $L$ copies 
$\left\{N^{[j]}(t)\right\}$ of the counting
process $N(t)$, where $N^{[j]}(t)$ is represented by the feature vector \mbox{$
	\mathbf{X}_{j} = \left[
	t_{1}^{[j]}, \ldots, t_{N^{[j]}}^{[j]}; N^{[j]}
	\right]
	$} with $N^{[j]}=N^{[j]}(T)$. The local martingale decomposition in~\eqref{eqn:2.9} for 
$N^{[j]}(t)$ reads
\begin{equation*}
	dN^{[j]}(t) = 
	\lambda(t)dt + dM^{[j]}(t),
	\enspace j=1,\ldots,L
	\,.
\end{equation*}
This gives the analogous decomposition for the aggregated counting process, i.e., we
have
\begin{equation}
	d\overline{N}_{L}(t) = \lambda(t)dt + d\overline{M}_{L}(t)
	\label{eqn:4.11}\,,
\end{equation}
where 
\begin{equation*}
	\begin{split}
		& d\overline{N}_{L}(t) = \frac{1}{L}	\sum_{j=1}^{L} dN^{[j]}(t) \,,\\
		& d\overline{M}_{L}(t) = \frac{1}{L}	\sum_{j=1}^{L} dM^{[j]}(t) \,.
	\end{split}
\end{equation*}

It is important to note that the aggregated residual process $d\overline{M}_{L}(t)$ 
meets all the properties listed in Appendix~A. Hence, $\E{d\overline{M}_{L}(t)}=0$ and
the properties in~\eqref{eqn:a6} and~\eqref{eqn:a7} are as follows
\begin{equation}
	\begin{aligned}
		\Var{d\overline{M}_{L}(t)} &= 
		\frac{1}{L} \lambda(t)dt
		\,,\\
		\Var{\int_{0}^{T} g(u) d\overline{M}_{L}(u)} &= 
		\frac{1}{L} \int_{0}^{T} g^{2}(u) \lambda(u)du
		\,.
	\end{aligned}
	\label{eqn:4.12}
\end{equation}
The single-sample kernel estimate of $\lambda(t)$ is as in~\eqref{eqn:4.9},
whereas its aggregated version takes the form
\begin{equation}
	\widehat{\lambda}(t) = 
	\frac{1}{L} \sum_{j=1}^{L} \widehat{\lambda}^{[j]}(t)
	\label{eqn:4.13}\,.
\end{equation}
This due to~\eqref{eqn:4.11} can be written in the convenient stochastic integral form
\begin{equation}
	\widehat{\lambda}(t) = 
	\int_{0}^{T} K_{h}\left(t-s\right) d\overline{N}_{L}(s)
	\label{eqn:4.14}\,.
\end{equation}
Employing this identity along with~\eqref{eqn:4.11} and the aforementioned properties 
of $d\overline{M}_{L}(t)$ (see~\eqref{eqn:4.12})  yield the following identities
for the bias and the variance
of $\widehat{\lambda}(t)$
\begin{equation}
	\E{\widehat{\lambda}(t)} = 
	\int_{0}^{T} \frac{1}{h} K\left(\frac{t-s}{h}\right) \lambda(s)ds
	\label{eqn:4.15}\,,
\end{equation}
\begin{equation}
	\Var{\widehat{\lambda}(t)} = 
	\frac{1}{Lh}
	\int_{0}^{T} \frac{1}{h} K^{2}\left(\frac{t-s}{h}\right) \lambda(s)ds
	\label{eqn:4.16}\,.
\end{equation}
These formulas and the standard analysis developed in the context of kernel 
estimates~\cite{wand1994kernel, grebpaw08} reveal that if
\begin{equation*}
	h(L) \to 0
	\enspace \text{and}
	\enspace Lh(L) \to \infty
\end{equation*}
then
\begin{equation}
	\widehat{\lambda}(t) \to \lambda(t)
	\enspace (P)
	\enspace \text{as}
	\enspace L \to \infty
	\label{eqn:4.17}
\end{equation}
at $t\in(0,T)$ where $\lambda(t)$ is continuous. This is the pointwise convergence that 
holds at interior points of $[0,T]$. It is
known~\cite{wand1994kernel, diggle1988equivalence} that the convergence fails at the
boundary points near $t=0$, \enspace$t=T$. This enforces us to confine the required
uniform convergence to the interval $[\epsilon, T-\epsilon]$ for arbitrarily
small $\epsilon>0$. Yet another option
is to introduce the boundary modified 
kernels~\cite{diggle1988equivalence, cattaneo2020simple} that are able to restore the 
convergence property at the boundary points. 
The following lemma gives the sufficient conditions for the uniform
convergence property of the estimate $\widehat{\lambda}(t)$ in~\eqref{eqn:4.13}.
\begin{lemma}
	\label{lemma:8}
	Let~$\lambda(t)$ be  Lipschitz continuous on~$[0,\infty)$. 
	Let the kernel function $K(t)$ be Lipschitz continuous on $[-1,1]$.
	Suppose 
	that
	\begin{equation}
		h(L) \to 0
		\enspace \text{and}
		\enspace Lh^{3}(L) \to \infty
		\enspace \text{as}
		\enspace L \to \infty
		\label{eqn:4.17a}\,.
	\end{equation}
	Then for arbitrarily small $\epsilon>0$
	\begin{equation}
		\sup_{t\in\left[\epsilon, T-\epsilon\right]} \left|
		\widehat{\lambda}(t)-\lambda(t)
		\right| \to 0
		\enspace (P)
		\enspace \text{as}
		\enspace L \to \infty
		\label{eqn:4.18}\,.
	\end{equation}
\end{lemma}

It is worth noting that the uniform convergence holds under the condition 
\mbox{$Lh^{3}(L)\to\infty$}. This is the stronger restriction than the one required
for the pointwise
convergence, where one needs that \mbox{$Lh(L)\to\infty$}. We conjecture 
that~\eqref{eqn:4.17a} can be replaced by the weaker condition 
\mbox{$Lh(L)/\log(L)\to\infty$} This is the case for the uniform convergence of the 
kernel density estimate where advanced tools from the empirical processes theory have 
been utilized~\cite{gine2002rates,masry1996multivariate}. Our proof is based on more 
elementary techniques. The proof of Lemma~\ref{lemma:8} is given in Appendix~D. The
result of Lemma~\ref{lemma:8} applies directly to the shape densities and by using 
Theorem~\ref{theorem:5} and Theorem~\ref{theorem:6} we can formulate the following Bayes 
risk consistency result for the kernel classifier.
\begin{theorem}
	\label{theorem:7}
	Let the class intensities $\lambda_{1}(t)$, $\lambda_{2}(t)$
	satisfy the conditions of Lemma~\ref{lemma:7}. 
	If~\eqref{eqn:4.17a} holds, then the kernel classification rule is Bayes risk 
	consistent, i.e., we have
	\begin{equation*}
		\mathbf{R}_{L,T} \to \mathbf{R}_{T}^{\star}
		\enspace (P)
	\end{equation*}
	as $L\to\infty$.
\end{theorem}
The convergence in Theorem~\ref{theorem:7} is an important property of the kernel 
classifier. Nevertheless, the issue of the rate of convergence would also be essential. 
This question is left for further research.

The selection of the bandwidth $h$ is the most important issue in determining the 
the finite sample accuracy of the kernel classification rule. The standard analysis
applied to the
expression in~\eqref{eqn:4.14} and~\eqref{eqn:4.15} shows that if $\lambda(t)$ has two 
continuous derivatives for $t\in(0,T)$ then
\begin{equation*}
	\E{\widehat{\lambda}(t)} = \lambda(t) + \mathcal{O}\left(h^{2}\right)
\end{equation*}
and
\begin{equation*}
	\Var{\widehat{\lambda}(t)} = \mathcal{O}\left(\frac{1}{Lh}\right)
	\,.
\end{equation*}
This leads to the following asymptotical formula for the mean squared error
\begin{equation*}
	\E{\widehat{\lambda}(t)-\lambda(t)}^{2} = 
	\mathcal{O}\left(\frac{1}{Lh}\right) + \mathcal{O}\left(h^{4}\right),
	\enspace t\in(0,T)
	\,.
\end{equation*}
The minimum of the error yields the asymptotically optimal choice of the bandwidth, i.e., 
$h^{\star}=cL^{-1/5}$ for some positive constant $c$. This is the asymptotically 
optimal choice of $h$ that optimizes the kernel intensity estimate. 
An optimal bandwidth for the kernel
classifier may be quite different as it is seen from the restriction
in~\eqref{eqn:4.17a}. See also~\cite{audibert2007fast}
for the general theory of plug-in nonparametric classifiers.\\
In practical applications one can specify the bandwidth using some resampling techniques
like cross-validation~\cite{diggle1988equivalence, wand1994kernel}. In our experimental
studies we choose separate bandwidth for each class. This is done by finding the 
maximum of the cross-validated log-likelihood of the kernel estimate of the shape 
densities. 
Hence, let $\widehat{p}_{i}(t;h)$ be the kernel estimate in~\eqref{eqn:4.10} specified 
by the bandwidth $h$. Then, the likelihood function of $\widehat{p}_{i}(t;h)$
specified by test data is given by
\begin{equation}
	\mathbf{CV}(h) = \prod_{l=1}^{p}\prod_{r=1}^{N^{[l]}}
	\widetilde{p}_{i}(t_{r}^{[l]};h)
	\label{eqn:4.19} \,,
\end{equation}
where $t_{r}^{[l]}$ represents the $r$-th observation of the $l$-th test sample.
We use the test sample of size $q$ (per class). Also $\widetilde{p}_{i}(t;h)$ is
the version of $\widehat{p}_{i}(t;h)$ in~\eqref{eqn:4.10} determined from
the $L_{i}-q$ size training set. Then, the bandwidth is selected as the one that
maximizes $\mathbf{CV}(h)$ in~\eqref{eqn:4.19}. This is equivalent to the following
choice
\begin{equation*}
	\widehat{h}_i =
	\argmax_h \sum_{l=1}^{p}\sum_{r=1}^{N^{[l]}} 
	\log \left(\widetilde{p}_{i}(t_{r}^{[l]};h)\right)
	\,.
\end{equation*}

\section{Simulation Results}
\label{sec:05_simulations}

\noindent In order to assess the proposed methodology, we conduct a simulated data study.
We limit the scope of our experiments to time-dependent intensity functions defined
in~\eqref{eqn:sim1}, and use these in simulations in order to gain insight into the
behavior of $\mathbf{R}_{L,T}$ with respect to the training set size $L$ and the 
observation window size $T$. 

In all experiments the kernel classifier is given by~\eqref{eqn:4.3} with the estimated 
$\widehat{\tau}_i$,\enspace $\widehat{p}_i(t;h)$ specified by~\eqref{eqn:4.1} 
and~\eqref{eqn:4.10}, respectively. The Gaussian kernel is employed,
whereas the bandwidth is selected by the log-likelihood method in~\eqref{eqn:4.19}. 
When selecting the bandwidth, we consider a grid of ten evenly logarithmically spaced 
points $h\in\langle 10^{-1}, 10^{1}\rangle$. Additionally, we employ a 5-fold cross 
validation in order to avoid biasing the selected bandwidth $\widehat{h}_i$ with the test 
data. Finally, we denote $\E{\mathbf{R}_{L,T}}$ as an empirically evaluated risk 
averaged over ten simulation runs with a testing set size of~$10^{4}$.

We shall focus on the simulation results obtained for the intensity 
function specified by~\eqref{eqn:sim1}. Unless noted otherwise, we refer to the 
intensity function pair parametrized by $\phi_{1}=\pi/16$ and $\phi_{2}=\pi/4$.

Figure~\ref{fig:05_simulations/sims/time} depicts the average risk versus $T$ for the 
size of training data ranging from $L=10$ to $L=200$. The Bayes risk 
$\mathbf{R}^{*}_{T}$ is also plotted for comparison. The convergence of  
$\mathbb{E}[\mathbf{R}_{L,T}]$ to zero analogous as it was observed for the Bayes risk
(see \figurename~\ref{fig:03_bayes_bounds/risk/simulations}) is seen.
Also the small value
of the difference $\mathbb{E}[\mathbf{R}_{L,T}] - \mathbf{R}^{*}_{T}$ for all  $T$ 
should be noted. We also observe the small variability of the risk with respect to the 
training data size $L$. The vertical dashed line at $T=10$ denotes the simulation space 
slice in subsequent analysis, i.e., with the value of $T$ fixed.

\begin{figure}[!t]
	\includegraphics[width=\figurewidth]{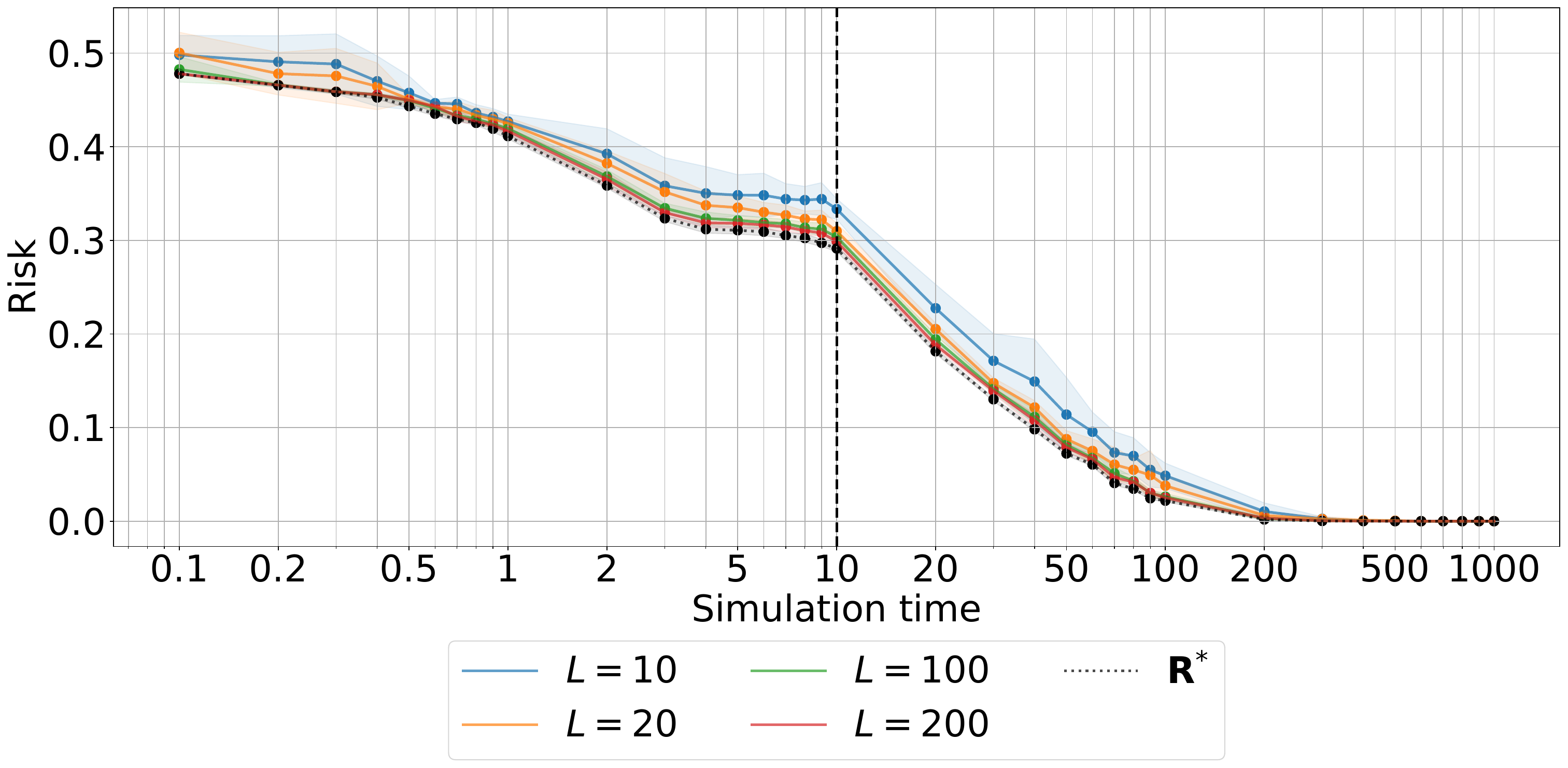}
	\centering
	\caption{
		The average risk $\E{\mathbf{R}_{L,T}}$ versus $T$ for different values of $L$. 
		The vertical dashed line at $T=10$ denotes the simulation space slice presented 
		in \figurename~\ref{fig:05_simulations/sims/bandwidth}b-\ref{fig:05_simulations/sims/examples}.
	}
	\label{fig:05_simulations/sims/time}
\end{figure}

Next, we analyze the value of the optimal bandwidth selected according to the 
log-likelihood method versus $T$. For brevity, in 
\figurename~\ref{fig:05_simulations/sims/bandwidth}a we show only the results for 
$\widehat{h}_{1}$, noting that the curves obtained for $\widehat{h}_{2}$ are analogous. 
We observe an increase in $\widehat{h}_{1}$ with $T$, which aligns with the notion that 
as the observation window increases, the distribution of events in time becomes 
sparser, yielding the larger bandwidth. On the other hand, the obtained results also show 
that $h(L)\to0$ as $L$ increases. Another way to view this property is to analyze the 
model log-likelihood versus $h$ for fixed $T$ 
(\figurename~\ref{fig:05_simulations/sims/bandwidth}b). 

\begin{figure}[!t]
	\includegraphics[width=\figurewidth]{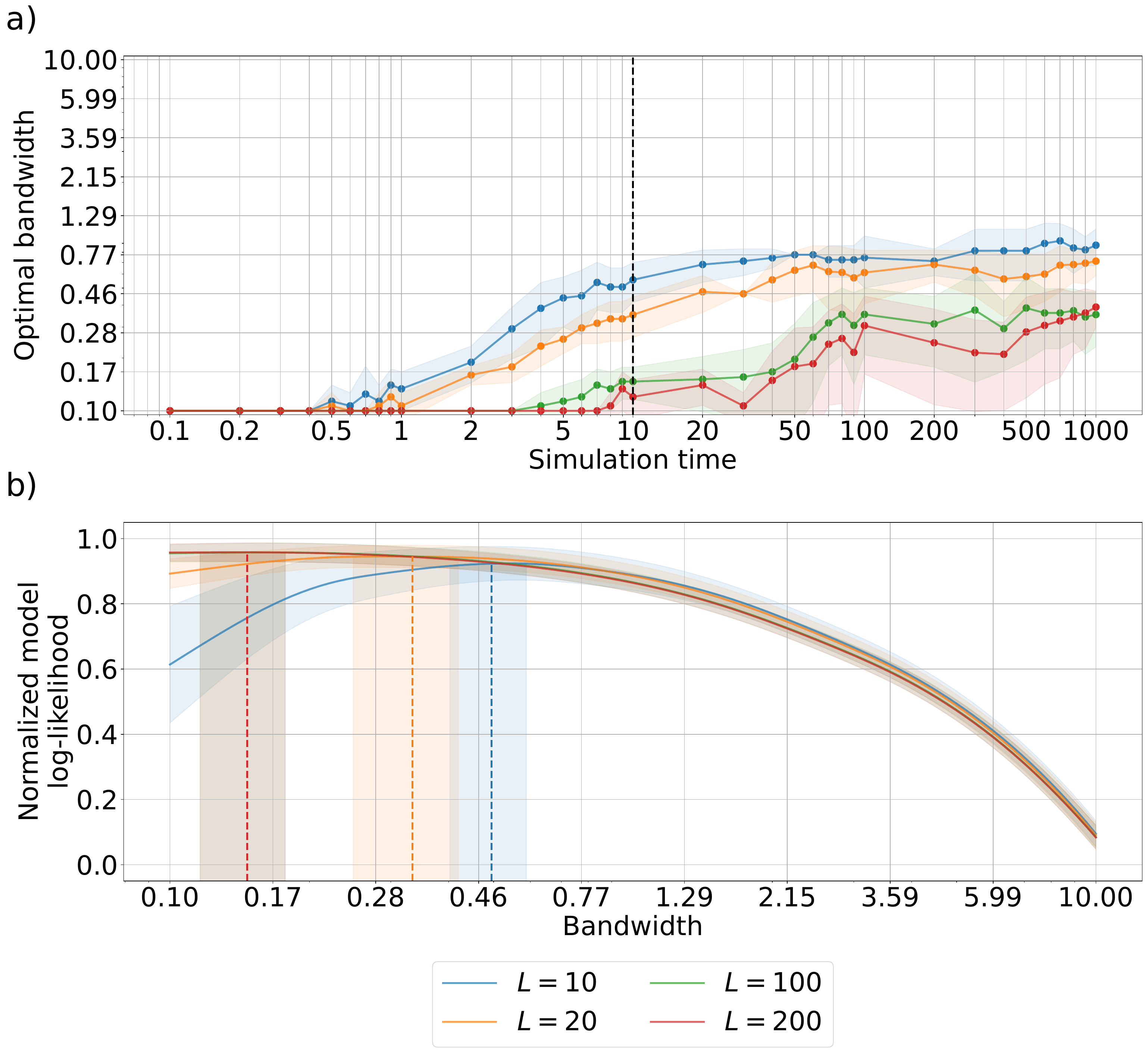}
	\centering
	\caption{
		a)~The average empirical bandwidth $\E{\widehat{h}_{1}}$ versus $T$ for different 
		values of $L$. The vertical dashed 
		line at $T=10$ denotes the simulation space slice presented in the lower 
		subfigure.
		b)~The average normalized model log-likelihood on test data versus $h$ for 
		different values of $L$. The vertical dashed lines denote function maxima. Note 
		that the curves for $L=100$ and $L=200$ overlap.
	}
	\label{fig:05_simulations/sims/bandwidth}
\end{figure}

Finally, \figurename~\ref{fig:05_simulations/sims/examples} shows the convergence of the 
empirical kernel rule risk to the Bayes risk for different values of the intensity 
function pair parameters $\phi_{1}$,\enspace $\phi_{2}$ versus $L$ and fixed $T=10$. 
Clearly, as the difficulty of the problem increases, i.e., when the two intensity
functions become more similar to one another, the rate of convergence decreases.
Also note that the Bayes risk is higher for more difficult
classification problems.

\begin{figure}[!t]
	\includegraphics[width=\figurewidth]{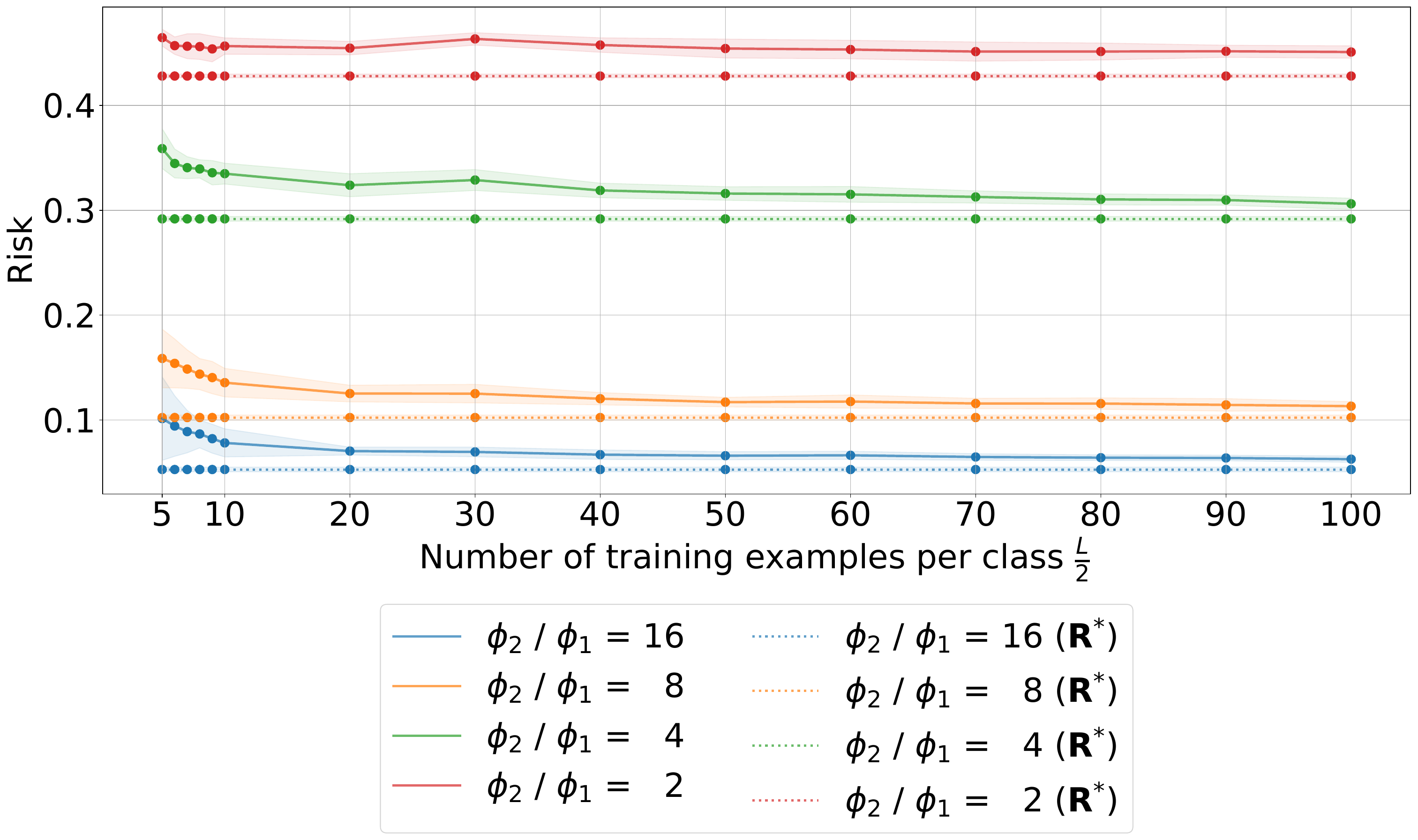}
	\centering
	\caption{
		The average risk $\mathbb{E}[\mathbf{R}_{L,T}]$ versus $L$ at given $T$ for 
		different values intensity function pairs parametrized by $\phi_{1}$,\enspace 
		$\phi_{2}$. The horizontal dashed lines denote estimated Bayes risk 
		$\mathbf{R}^{*}_{T}$ for the associated intensity function pair.
	}
	\label{fig:05_simulations/sims/examples}
\end{figure}

Let us briefly examine a counter-example when the proposed algorithm fails to converge.
Consider the following Gaussian type intensity function
\begin{equation}
	\lambda(t; a,b) = a \exp \left[-b\left(t-0.5\right)^{2}\right]
	\label{eqn:5.1} \,,
\end{equation}
which does not satisfy the assumptions $\assumption{1}$ and $\assumption{2}$.
While the intensity function has an infinite support, in practice it is extremely 
unlikely for events to occur outside of some narrow time interval. 
In \figurename~\ref{fig:05_simulations/sims/failure} we consider the classification
problem with \mbox{$\lambda_{1}(t)=\lambda\left(t; 300, 20\right)$} and
\mbox{$\lambda_{2}(t)=\lambda\left(t; 600, 40\right)$}.  For such specified intensity
functions we can evaluate that $\int_{0}^{T} \lambda_1(t)dt$ $< 119$ and
$\int_{0}^{T} \lambda_2(t)dt$  $< 169$ for all $T$.
Hence, the average number of events from each class is finite and consequently
the condition {\bf A2} does not hold. Note that the empirical risk
does not converge to the Bayes risk that takes very small values for $T > 0.5$. 
Note also that the risk $\E{\mathbf{R}_{L,T}}$ is the smallest around the maximum
of~\eqref{eqn:5.1} at
$t=0.5$. Afterwards $\mathbf{R}_{L,T}$ slightly increases and reaches a plateau because 
no new events can be observed.

\begin{figure}[!t]
	\includegraphics[width=\figurewidth]{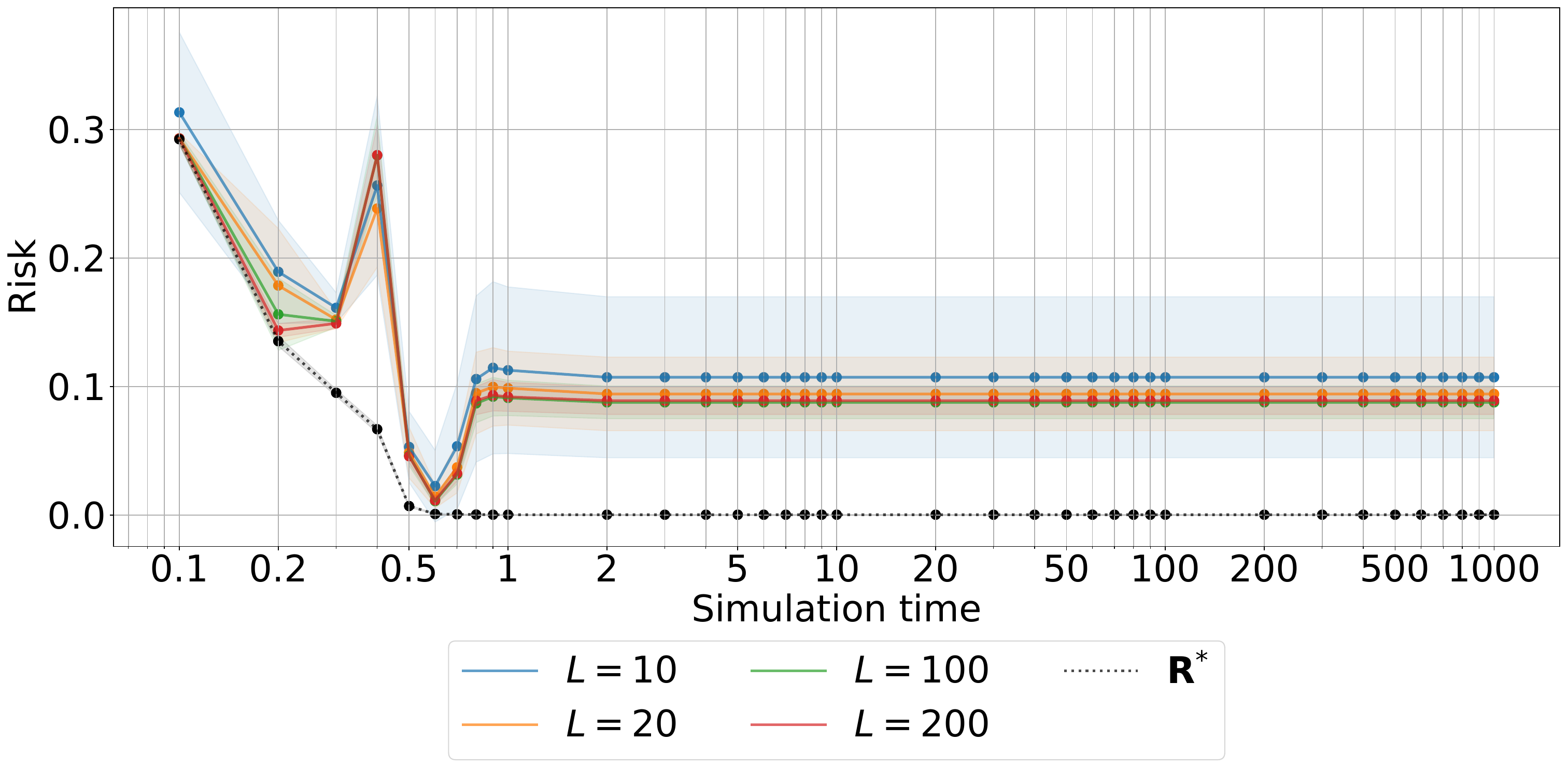}
	\centering
	\caption{
		The average risk $\E{\mathbf{R}_{L,T}}$ versus $T$ for different values of $L$ 
		for the Gaussian type intensity function.
	}
	\label{fig:05_simulations/sims/failure}
\end{figure}

\section{Concluding Remarks}
\label{sec:06_conclusions}

\noindent In this paper we have developed the rigorous asymptotic analysis for the 
classification
problem applied to spike trains data characterized by non-random intensity functions. The
optimal Bayes rule was derived and its finite and asymptotic (with respect to the length 
of the observation interval) properties were established. This includes the exponential 
bound for the Bayes risk. Our asymptotic theory is relied on the martingale 
representation of counting processes. We then introduced a general class of plug-in 
empirical classification rules and formulated the sufficient conditions for their 
convergence (as the amount of data grows) to the Bayes risk. This optimality property is 
confirmed and verified for the plug-in kernel classifier derived from the aggregated 
data. 

\noindent There are various ways to extend and generalize the results obtained in this 
paper. First of all, the log transformed version of the Bayes rule in~\eqref{eqn:2.4} 
holds for a general class of point processes such as the Hawkes self-excited 
process~\cite{lima023} and multivariate or marked point 
processes~\cite{daley2003introduction}. Hence, the extension of our results to this
type of
point processes is a natural topic for future research. 
The two-class classification problem studied in this paper has straightforward
generalization to the multi-class situation with the class labels denoted as \mbox{$
	\left\{\omega_{1}, \ldots, \omega_{c}\right\}
	$}. In fact, the Bayes rule in~\eqref{eqn:2.4} for the $c-$class 
classification problem reads as
\begin{equation*}
	\begin{split}
		\mathbf{X}\in\omega_{i}
		\enspace & \text{if}
		\enspace
		\sum_{s=1}^{N} \log\left(\frac{\lambda_{i}(t_s)}{\lambda_{k}(t_s)}\right) 
		\geq \gamma_{ik} \\
		& \text{for all}
		\enspace k=1, \ldots, c,
		\enspace k \neq i
	\end{split}
	\,,
\end{equation*}
where \mbox{$
	\gamma_{ik}=\int_{0}^{T}\left(
	\lambda_{i}(u)-\lambda_{k}(u)
	\right)du + \log\left(\pi_{k}/\pi_{i}\right)
	$}. Here $\left\{\lambda_{i}(t)\right\}$ are class intensity functions and 
$\left\{\pi_{i}\right\}$ are prior probabilities. 
Utilizing the martingale decomposition (see~\eqref{eqn:2.9}) for point processes would 
allow us to generalize our asymptotic results to the multi-class case. Also designing 
nonparametric plug-in classification rules with the desirable asymptotic optimality 
property would be of a great practical topic for further research.

\appendices

\section{}
\label{app:A}

The asymptotic theory of the classification problem examined in this paper is based on 
martingale methods. This appendix gives brief summary of the essential facts 
concerning the counting processes theory and their martingale representation, 
see~\cite{andersen2012statistical} for the full account of this theory. Hence, 
let~$N(t)$ be a spike train process which can be consider as a counting process of the 
occurrences in the interval~$\left[0,t\right]$ such that~$N(0)=0$. By~$dN(t)$ we 
denote the increment of~$N(t)$ over the small interval~$\left[t,t+dt\right)$ . The 
evolution of~$N(t)$ in time is completely characterized by the local intensity 
function~$\lambda(t)$. This is defined as
\begin{equation}
	\E{dN(t)|\mathbf{F}_t} = \lambda(t)dt
	\label{eqn:a1} \,,
\end{equation}
where~$\mathbf{F}_t$ denotes the history of~$N(t)$ in the interval~$\left[0,t\right)$. 
Note that~$\lambda(t)$ is generally random due to the dependence on the values 
of~$N(t)$ prior to the time~$t$. The formula in~\eqref{eqn:a1} implies that the 
residual process
\begin{equation}
	dM(t) = dN(t) - \lambda(t)dt
	\label{eqn:a2}
\end{equation}
satisfies the property
\begin{equation}
	\E{dM(t)|\mathbf{F}_t} = 0
	\label{eqn:a3} \,.
\end{equation}
This confirms the fact that the process
\begin{equation}
	M(t) = N(t) - \int_{0}^{t} \lambda(u) du
	\label{eqn:a4}
\end{equation}
is a zero mean local martingale.

The formula in~\eqref{eqn:a2} can be written as
\begin{equation}
	dN(t) = \lambda(t)dt + dM(t)
	\label{eqn:a5} \,.
\end{equation}
This can be viewed as the local signal plus noise decomposition of~$N(t)$. Moreover, 
the noise process~$dM(t)$ in~\eqref{eqn:a5} is a zero mean martingale that has 
uncorrelated but nonstationary increments~\cite{andersen2012statistical}. Based on 
these facts it can be shown that~$dM(t)$ has the following second order property
\begin{equation}
	\Var{dM(t)|\mathbf{F}_t} = \lambda(t)dt
	\label{eqn:a6} \,.
\end{equation}
Also \mbox{$\Var{dM(t)|\mathbf{F}_t}=\Var{dN(t)|\mathbf{F}_t}$}.

The fact that~$dM(t)$ has uncorrelated increments and that it reveals a piecewise 
constant sample paths allow us to define the stochastic Stieltjes type integral with 
respect to~$dM(t)$. Hence, let
\begin{equation*}
	\mathbf{I}(t) = \int_{0}^{t} g(u) dM(u)
\end{equation*}
define the stochastic integral of the measurable function~$g(t)$ with respect to the
increments of the martingale~$M(t)$. It is known~\cite{andersen2012statistical} that 
the martingale property is preserved under stochastic integration. Since 
\mbox{$\E{\mathbf{I}(t)}=0$} the integral~$\mathbf{I}(t)$ is a zero mean martingale 
with respect to the history of the counting process~$N(t)$. The variance 
of~$\mathbf{I}(t)$ is given by
\begin{equation}
	\Var{\mathbf{I}(t)|\mathbf{F}_t} =
	\int_{0}^{t} g^2(u) \lambda(u) du
	\label{eqn:a7} \,.
\end{equation}
The uncorrelated increments property of the martingale process allows us to establish 
the following generalized version of~\eqref{eqn:a7}
\begin{equation}
	\begin{split}
		& \Cov{\int_{0}^{t}g_1(u)dM(u),\,\int_{0}^{t}g_2(u)dM(u),\,|\mathbf{F}_t} \\
		& = \int_{0}^{t} g_1(u) g_2(u) \lambda(u) du
	\end{split}
	\label{eqn:a8} \,,
\end{equation}
where~$g_1(t)$,~$g_2(t)$ are measurable functions.

\section{}
\label{app:B}

To prove the results of this section we need the following elementary inequalities
\begin{equation}
	\frac{x-1}{x} \leq \log(x) \leq x-1,
	\enspace
	x>0
	\label{eqn:b1}\,.
\end{equation}
The tighter version of this inequality for $x\geq1$ reads as follows
\begin{equation}
	2\frac{x-1}{x+1} \leq \log(x) \leq \frac{x^2-1}{2x},
	\enspace
	x\geq1
	\label{eqn:b2}\,.
\end{equation}

\begin{proof}[Proof of Lemma~\ref{lemma:1}]
	Let $\mathbf{X}\in\omega_{1}$. Then, the formula for the threshold value 
	$\alpha_{T}$ in~\eqref{eqn:3.3} becomes
	\begin{equation}
		\alpha_{T} = 
		\tau_{1} - \tau_{2} + \tau_{1}\log\left(\frac{\tau_{2}}{\tau_{1}}\right) -
		\tau_{1}\int_{0}^{T}\log\left(\frac{p_{1}(t)}{p_{2}(t)}\right)p_{1}(t)dt
		\label{eqn:b3}\,.
	\end{equation}
	Here \mbox{$
		\mathbf{K}_{T}\left(p_{1}\parallel p_{2}\right) =
		\int_{0}^{T}\log\left(\frac{p_{1}(t)}{p_{2}(t)}\right)p_{1}(t)dt
		$} is the Kullback-Leibler divergence between the densities $p_{1}(t)$ and $p_{2}(t)$. 
	Then, by virtue of~\eqref{eqn:b1} we have
	\begin{equation*}
		\begin{split}
			\alpha_{T} & \leq
			\tau_{1} - \tau_{2} + \tau_{1}\left\{\frac{\tau_{2}}{\tau_{1}}-1\right\}-
			\tau_{1}\mathbf{K}_T(p_{1} \parallel p_{2}) \\
			& = -\tau_{1}\mathbf{K}_T(p_{1} \parallel p_{2})
		\end{split}
		\,.
	\end{equation*}
	As $\mathbf{K}_T(p_{1} \parallel p_{2})\geq0$ we conclude that $\alpha_{T}\leq0$. 
	Concerning the lower bound for $\alpha_{T}$ in~\eqref{eqn:b3} we again 
	use~\eqref{eqn:b1}. Hence,
	\begin{equation*}
		\begin{split}
			\alpha_{T} & \geq
			\tau_{1} - \tau_{2} + \tau_{1}\left\{1-\frac{\tau_{1}}{\tau_{2}}\right\}-
			\tau_{1}\mathbf{K}_T(p_{1} \parallel p_{2}) \\
			& = \frac{\left(\tau_{1}-\tau_{2}\right)^2}{\tau_{2}}-
			\tau_{1}\mathbf{K}_T(p_{1} \parallel p_{2})
		\end{split}
		\,.
	\end{equation*}
	This confirms the inequalities in Lemma~\ref{lemma:1}~(a). The case when 
	$\mathbf{X}\in\omega_{2}$ can be proved in 	the analogous way by noting that 
	$\alpha_{T}$ is now equal to
	\begin{equation*}
		\alpha_{T}=
		\tau_{1} - \tau_{2} + \tau_{2}
		\log\left(\frac{\tau_{2}}{\tau_{1}}\right)+
		\tau_{2}\mathbf{K}_T(p_{2} \parallel p_{1})
		\,.
	\end{equation*}
	Then, the application of~\eqref{eqn:b1} gives the result in Lemma~\ref{lemma:1}~(b).
\end{proof}

\begin{proof}[Proof of Lemma~\ref{lemma:2}]
	The result of Lemma~\ref{lemma:2} is implied by the straightforward application of 
	the identity in~\eqref{eqn:a7} to the stochastic integral
	\begin{equation*}
		\int_{0}^{T}\log\left(\frac{\lambda_{1}(t)}{\lambda_{2}(t)}\right)dM(t)
		\,.
	\end{equation*}
	Here $M(t)$ is the local martingale corresponding to the intensity 
	$\lambda_{1}(t)$ or $\lambda_{2}(t)$ depending whether $\mathbf{X}\in\omega_{1}$ 
	or $\mathbf{X}\in\omega_{2}$, respectively.
\end{proof}

\begin{proof}[Proof of Lemma~\ref{lemma:5}]
	The result in~\eqref{eqn:3.15} of Lemma~\ref{lemma:5} is the version of Theorem~5 
	in~\cite{le2021exponential} that says that under the conditions (a) and (b) of 
	Lemma~\ref{lemma:5} we have 
	\begin{equation}
		\mathbf{P}\left(\left|U_{T}\right|\geq\epsilon\right) \leq
		2\exp\left[
		-\frac{v_{T}}{u_{T}^2}J\left(\epsilon\frac{u_{T}}{v_{T}}\right)
		\right]
		\label{eqn:b4}\,,
	\end{equation}
	where $J(x)=(1+x)\log\left(1+x\right)-x$,\enspace$x\geq0$. Using the inequalities
	in~\eqref{eqn:b2} we can easily obtain that
	\begin{equation*}
		\frac{x^2}{2+x} \leq J(x) \leq \frac{x^2}{2}
		\,.
	\end{equation*}
	The application of the above lower bound in~\eqref{eqn:b4} leads to the version 
	of~\eqref{eqn:b4} given in~\eqref{eqn:3.15}.
\end{proof}

\begin{proof}[Proof of Theorem~\ref{theorem:3}]
	We will prove the result in~\eqref{eqn:3.21}. This clearly implies the
	convergence $\mathbf{R}_{T}^{*}\to0$ as $T\to\infty$. By virtue of~\eqref{eqn:2.8}
	it suffices to consider the probability of misclassification
	$\mathbf{P}\left(W_{T}(\mathbf{X})\geq\eta_{T}|\mathbf{X}\in\omega_{2}\right)$. As it
	has been observed in~\eqref{eqn:3.19} this probability is equivalent 
	to the following probability
	\begin{equation}
		\mathbf{P}\left(
		\frac{1}{T}U_{T}(\mathbf{X})\geq\frac{1}{T}\alpha_{T} +
		\frac{1}{T}\kappa | \mathbf{X}\in\omega_{2}
		\right)
		\label{eqn:b7}\,,
	\end{equation}
	where $\kappa=\log\left(\pi_{2}/\pi_{1}\right)$. Since $\mathbf{X}\in\omega_{2}$ 
	then
	\begin{equation}
		\alpha_{T} = 
		\tau_{1} - \tau_{2} + \tau_{2} \log\left(\frac{\tau_{2}}{\tau_{1}}\right) +
		\tau_{2} \mathbf{K}_T(p_{2} \parallel p_{1})
		\label{eqn:b8}\,.
	\end{equation}
	Then, by the Chebyshev inequality the probability in~\eqref{eqn:b7} is bounded by
	\begin{equation}
		b_{T}\frac{1}{T}
		\label{eqn:b9}\,,
	\end{equation}
	where
	\begin{equation}
		b_{T} = \frac{
			\Var{\frac{1}{\sqrt{T}}U_{T}(\mathbf{X})}	
		}{
			\left[T^{-1}\alpha_{T}+T^{-1}\kappa\right]^2
		}
		\label{eqn:b10}\,.
	\end{equation}
	By the assumptions $\assumption{1}$ and $\assumption{2}$ we have
	\begin{equation}
		\overline{\lim}_{T\to\infty} \alpha_{T} / T \leq
		d \log\left(\frac{C}{\delta}\right)
		\label{eqn:b6-1}
	\end{equation}
	and also
	\begin{equation}
		\underline{\lim}_{T\to\infty} \alpha_{T} / T \geq
		d \log\left(\frac{\delta}{C}\right)
		\label{eqn:b6-2} \,.
	\end{equation}
	Then by the result of Lemma~\ref{lemma:4} and~\eqref{eqn:b6-2}, we get
	\begin{equation}
		\overline{\lim}_{T\to\infty} b_{T} \leq
		\frac{
			d \log^{2}\left(\frac{C}{\delta}\right)
		}{
			d^{2} \log^{2}\left(\frac{\delta}{C}\right)
		} = \frac{1}{d} \left(
		\frac{\log(C/\delta)}{\log(\delta/C)}			
		\right)^{2}
		\label{eqn:b11} \,.
	\end{equation}
	Hence, the probability of misclassification \mbox{$
		\mathbf{P}\left(W_{T}(\mathbf{X})\geq\eta_{T}|\mathbf{X}\in\omega_{2}\right)
		$} is bounded by $b_{T}/T$,	where the superior limit of $b_{T}$ is given in~\eqref{eqn:b11}.
	In the analogous way one can show that the probability of misclassification \mbox{$
		\mathbf{P}\left(W_{T}(\mathbf{X})<\eta_{T}|\mathbf{X}\in\omega_{1}\right)
		$} is bounded by $a_{T}/T$, where the superior limit of $a_{T}$ is also given
	by ~\eqref{eqn:b11}.
	This concludes the proof of Theorem~\ref{theorem:3}.
\end{proof}

\begin{proof}[Proof of Theorem~\ref{theorem:4}]
	Consider again 
	$\mathbf{P}\left(W_{T}(\mathbf{X})\geq\eta_{T}|\mathbf{X}\in\omega_{2}\right)$ or
	equivalently the probability in~\eqref{eqn:b7}. We wish to use the exponential 
	inequality in~\eqref{eqn:3.17} of Lemma~\ref{lemma:6}. Then, the probability 
	in~\eqref{eqn:b7} is bounded by
	\begin{equation}
		\exp\left[
		-T \frac{\epsilon_{T}^{2}}{2\theta_{T}+u\epsilon_{T}}
		\right]
		\label{eqn:b12}\,,
	\end{equation}
	where $u=\log\left(\frac{C}{\delta}\right)$ characterizes the assumption 
	$\assumption{1}$, \mbox{$\epsilon_{T}=\frac{1}{T}\alpha_{T}+\frac{1}{T}\kappa$}, 
	and \mbox{$\theta_{T}=\Var{\frac{1}{\sqrt{T}}U_{T}(\mathbf{X})}$}. This defines 
	the exponential factor
	\begin{equation*}
		B_{T} = \frac{\epsilon_{T}^{2}}{2\theta_{T}+u\epsilon_{T}}
		\,.
	\end{equation*}
	Owing to Lemma~\ref{lemma:4}, \eqref{eqn:b6-1} and~\eqref{eqn:b6-2}, the limit inferior
	of~$B_{T}$ is not smaller than
	\begin{equation*}
		\begin{split}
			\underline{\lim}_{T\to\infty} B_{T} & \geq
			\frac{
				d^{2} \log^{2}(\delta/C)
			}{
				2d \log^{2}(C/\delta) + ud \log(C/\delta)
			} \\
			& = d \frac{1}{3} \left(
			\frac{\log(\delta/C)}{\log(C/\delta)}			
			\right)^{2}
		\end{split}
		\,.
	\end{equation*}	
	This combined with~\eqref{eqn:b12} gives the required bound. Since the analogous 
	analysis can be carried out for the probability of misclassification 
	$\mathbf{P}\left(W_{T}(\mathbf{X})<\eta_{T}|\mathbf{X}\in\omega_{1}\right)$ therefore
	the	proof of Theorem~\ref{theorem:4} has been completed.
\end{proof}

\section{}
\label{app:C}

\begin{proof}[Proof of Theorem~\ref{theorem:5}]
	The proof of Theorem~\ref{theorem:5} is in the spirit of the proof of Theorem~1
	in~\cite{greblicki1978asymptotically}. Hence, the consistency results established
	in~\eqref{eqn:4.6} and~\eqref{eqn:4.7} imply that for the selected $\delta>0$ there
	exists~$l_{0}$ such that for~$L > l_{0}$ and $\epsilon >0 $ we have
	\begin{equation}
		\begin{split}
			& \mathbf{P}\left(
			\left|
			\widehat{W}_{L,T}(\mathbf{x}) - W_{T}(\mathbf{x})
			\right| < \epsilon
			\right) > 1 - \delta/2
			\,,\\
			& \mathbf{P}\left(
			\Big|
			\widehat{\eta}_{L,T} - \eta_{T}
			\Big| < \epsilon
			\right) > 1 - \delta/2
			\,.
		\end{split}
		\label{eqn:c4}
	\end{equation}
	Let \mbox{$\psi_{T}^{\star}(\mathbf{x})=\omega_{1}$}, i.e., we have
	\mbox{$W_{T}(\mathbf{x})>\eta_{T}$}. Then,
	\begin{equation*}
		\mathbf{P}\left(
		\widehat{\psi}_{L,T}(\mathbf{x}) = \psi_{T}^{\star}(\mathbf{x})
		\right) = \mathbf{P}\left(
		\widehat{W}_{L,T}(\mathbf{x}) > \widehat{\eta}_{L,T}
		\right)
		\,.
	\end{equation*}
	The right-hand side of this equality is not smaller than
	\begin{equation}
		\mathbf{P}\left(\left|
		\left(
		\widehat{W}_{L,T}(\mathbf{x}) - \widehat{\eta}_{L,T}
		\right) - \left(
		W_{T}(\mathbf{x}) - \eta_{T}
		\right)
		\right| < 2\epsilon \right)
		\label{eqn:c5}
	\end{equation}
	for \mbox{$0<\epsilon<\frac{1}{2}\left(W_{T}(\mathbf{x})-\eta_{T}\right)$}.
	Moreover,~\eqref{eqn:c5} is bounded from below by
	\begin{equation}
		\mathbf{P}\left(
		\left|
		\widehat{W}_{L,T}(\mathbf{x}) - W_{T}(\mathbf{x})
		\right| < \epsilon,\;
		\left|
		\widehat{\eta}_{L,T} - \eta_{T}
		\right| < \epsilon
		\right)
		\label{eqn:c6}\,.
	\end{equation}
	In turn by the elementary inequality \mbox{$
		\mathbf{P}(A \cap B) \geq \mathbf{P}(A) + \mathbf{P}(B) - 1
		$}, the lower bound for~\eqref{eqn:c6} is
	\begin{equation*}
		\mathbf{P}\left(
		\left|
		\widehat{W}_{L,T}(\mathbf{x}) - W_{T}(\mathbf{x})
		\right| < \epsilon
		\right)
		+
		\mathbf{P}\left(
		\left|
		\widehat{\eta}_{L,T} - \eta_{T}
		\right| < \epsilon
		\right) - 1
		\,.
	\end{equation*}
	Recalling~\eqref{eqn:c4} we have shown that for $L>l_{0}$
	\begin{equation*}
		\mathbf{P}\left(
		\widehat{\psi}_{L,T}(\mathbf{x}) = \psi_{T}^{\star}(\mathbf(x))
		\right) > 1-\delta
		\,.
	\end{equation*}
	Since we can choose an arbitrary small $\delta$, this confirms the claimed
	convergence.
\end{proof}

\begin{proof}[Proof of Lemma~\ref{lemma:7}]
	The proof will be based on the following version of Helly's 
	thorem~\cite{apostol1974mathematical} for the Stieltjes integral. 
	\begin{quoting}
		Let
		\begin{equation*}
			f_{L}(x) \to f(x)
			\enspace \text{uniformly on}
			\enspace [a,b]
			\enspace \text{as}
			\enspace L\to\infty
			\,.
		\end{equation*}
		If $g(x)$ is a function of bounded variation on $[a,b]$ then
		\begin{equation}
			\int_{a}^{b}f_{L}(x)dg(x) \to
			\int_{a}^{b}f(x)dg(x)
			\enspace \text{as}
			\enspace L\to\infty
			\label{eqn:c1}\,.
		\end{equation}
	\end{quoting}
	Consider the optimal decision function $W_{T}(\mathbf{x})$ in~\eqref{eqn:2.6} and 
	its empirical counterpart $\widehat{W}_{L,T}(\mathbf{x})$ in~\eqref{eqn:4.3}. Then, 
	we can write (see~\eqref{eqn:2.4a})
	\begin{equation}
		\begin{split}
			& \widehat{W}_{L,T}(\mathbf{x}) - W_{L,T}(\mathbf{x}) \\
			& = 
			\int_{0}^{T}\left[
			\log\left(\frac{\widehat{p}_{1}(t)}{\widehat{p}_{2}(t)}\right) - 
			\log\left(\frac{p_{1}(t)}{p_{2}(t)}\right)
			\right] dN(t)
		\end{split}
		\label{eqn:c2}\,.
	\end{equation}
	We wish to prove that \mbox{$
		\left|\widehat{W}_{L,T}(\mathbf{x}) - W_{L,T}(\mathbf{x})\right|
		$}\enspace$(P)$ as $L\to\infty$. Owing to Helly's theorem it suffices to show that
	\begin{equation}
		\begin{split}
			& \left|
			\log\left(\frac{\widehat{p}_{1}(t)}{\widehat{p}_{2}(t)}\right) - 
			\log\left(\frac{p_{1}(t)}{p_{2}(t)}\right)
			\right| \to 0
			\enspace (P) \\
			& \text{uniformly on}
			\enspace [0,T]
			\enspace \text{as}
			\enspace L\to\infty
		\end{split}
		\label{eqn:c3}\,.		
	\end{equation}
	Observe that the left-hand side of~\eqref{eqn:c3} is equal to \mbox{$
		\left|
		\log\left(
		\frac{\widehat{p}_{1}(t)p_{2}(t)}{\widehat{p}_{2}(t)p_{1}(t)}
		\right)
		\right|
		$}. Then, using~\eqref{eqn:b1} this is bounded by
	\begin{equation*}
		\begin{split}
			& \left|
			\frac{
				\widehat{p}_{1}(t)p_{2}(t)-\widehat{p}_{2}(t)p_{1}(t)
			}{
				\widehat{p}_{2}(t)p_{1}(t)
			}
			\right| \\
			& = \left|
			\frac{
				\left(
				\widehat{p}_{1}(t)-p_{1}(t)
				\right)p_{2}(t) +
				\left(
				p_{2}(t)-\widehat{p}_{2}(t)
				\right)p_{1}(t)
			}{
				\left(
				\widehat{p}_{2}(t)-p_{2}(t)
				\right)p_{1}(t) + p_{1}(t)p_{2}(t)
			}
			\right|
		\end{split}
		\,.
	\end{equation*}
	This is not greater than
	\begin{equation*}
		\frac{
			\left|
			\widehat{p}_{1}(t)-p_{1}(t)
			\right|p_{2}(t) +
			\left|
			\widehat{p}_{2}(t)-p_{2}(t)
			\right|p_{1}(t)
		}{
			p_{1}(t)p_{2}(t)
		}
		\,.
	\end{equation*}
	By the assumption~$\assumption{1}$ limited to the interval $[0,T]$ and the fact 
	that \mbox{$
		p_{i}(t) = \lambda_{i}(t)/\tau_{i}
		$} the above expression does not exceed
	\begin{equation*}
		\left(\frac{C}{\delta}\right)^{3}T\left\{
		\left|
		\widehat{p}_{1}(t)-p_{1}(t)
		\right| + 
		\left|
		\widehat{p}_{2}(t)-p_{2}(t)
		\right|
		\right\}
		\,.
	\end{equation*}
	This by recalling the assumption in~\eqref{eqn:4.4} proves~\eqref{eqn:c3}. The 
	proof of Lemma~\ref{lemma:7} has been completed.
\end{proof}

\section{}
\label{app:D}

\begin{proof}[Proof of Lemma~\ref{lemma:8}]
	We wish to show that
	\begin{equation}
		\sup_{t\in T_{\epsilon}} \left|
		\widehat{\lambda}(t)-\lambda(t)
		\right| \to 0
		\enspace (P)
		\enspace \text{as}
		\enspace L \to \infty
		\label{eqn:d1}\,.
	\end{equation}
	where $T_{\epsilon}=[\epsilon,T-\epsilon]$ for small $\epsilon>0$. We begin with the 
	standard bounding into the variance and bias terms
	\begin{equation}
		\left|\widehat{\lambda}(t)-\lambda(t)\right| \leq
		\left|\widehat{\lambda}(t)-\E{\widehat{\lambda}(t)}\right| + 
		\left|\E{\widehat{\lambda}(t)}-\lambda(t)\right|
		\label{eqn:d2}\,.
	\end{equation}
	Owing to~\eqref{eqn:4.15} the bias term is equal to
	\begin{equation*}
		\E{\widehat{\lambda}(t)}-\lambda(t) = \int_{-t/h}^{(T-t)/h}
		K(s)\lambda\left(t+hs\right)ds-\lambda(t)
	\end{equation*}
	for $t\in T_{\epsilon}$. Since $K(t)$ and $\lambda(t)$ are positive and $K(t)$ is a 
	density function supported on $[-1,1]$ then we have
	\begin{equation*}
		\begin{split}
			\left|\E{\widehat{\lambda}(t)}-\lambda(t)\right| 
			& = \int_{-1}^{1} K(s)\left|\lambda\left(t+hs\right)-\lambda(t)\right|ds \\
			& \leq M_{\lambda}h \int_{-1}^{1} K(s)|s|ds
		\end{split}
	\end{equation*}
	uniformly in $t\in T_{\epsilon}$, where $M_{\lambda}$ is the Lipschitz
	constant of $\lambda(t)$.\\
	Let us consider the stochastic part in~\eqref{eqn:d2}. As the interval $T_{\epsilon}$ 
	is compact, one can define a finite partition of $T_{\epsilon}$ into disjoint 
	equal size intervals, i.e., \mbox{$
		T_{\epsilon} = \bigcup\limits_{j=1}^{q(L)} \mathbf{U}_{j}
		$}, where the size of $\mathbf{U}_{j}$ is denoted as $\Delta(L)$. Clearly the number 
	of intervals is of order \mbox{$T_{\epsilon}/\Delta(L)$}. Let 
	$u_{j}\in\mathbf{U}_{j}$ be the middle point of $\mathbf{U}_{j}$. Then, the uniform 
	norm of the stochastic term in~\eqref{eqn:d2} can be bounded as follows
	\begin{equation}
		\begin{split}
			& \sup_{t\in T_{\epsilon}}\left|
			\widehat{\lambda}(t)-\E{\widehat{\lambda}(t)}
			\right| \\
			& \leq \max_{1\leq j\leq q(L)} \sup_{t\in T_{\epsilon}\cap \mathbf{U}_j}
			\left|\widehat{\lambda}(t)-\widehat{\lambda}(u_{j})\right| \\
			& + \max_{1\leq j\leq q(L)} \sup_{t\in T_{\epsilon}\cap \mathbf{U}_j}
			\left|\E{\widehat{\lambda}(t)}-\E{\widehat{\lambda}(u_{j})}\right| \\
			& + \max_{1\leq j\leq q(L)} \left|
			\widehat{\lambda}(u_{j}) - \E{\widehat{\lambda}(u_{j})}
			\right| \\
			& = A_{1} + A_{2} + A_{3}.
		\end{split}
		\label{eqn:d3} \,.
	\end{equation}
	Consider first the term $A_{1}$. By virtue of~\eqref{eqn:4.14} we have
	\begin{equation*}
		\widehat{\lambda}(t)-\widehat{\lambda}(u_{j}) = \int_{0}^{T} \left[
		K_{h}(t-s)-K_{h}(u_{j}-s)
		\right] d\overline{N}_{L}(s)
	\end{equation*}
	for $t,u_{j}\in\mathbf{U}_{j}$. Noting that \mbox{$|t-u_{j}|\leq\Delta(L)$} and using 
	the fact that $K(t)$ is Lipschitz we get
	\begin{equation}
		\left|\widehat{\lambda}(t)-\widehat{\lambda}(u_{j})\right| \leq
		M_{K}\frac{\Delta(L)}{h^{2}} \int_{0}^{T} d\overline{N}_{L}(s)
		\label{eqn:d4} \,.
	\end{equation}
	Note that \mbox{$
		\int_{0}^{T} d\overline{N}_{L}(s) = \overline{N}_{L}(T)
		$} and we know, see~\eqref{eqn:4.11}, that \mbox{$
		\E{\overline{N}_{L}(T)} = \int_{0}^{T} \lambda(t)dt
		$} and \mbox{$
		\Var{\overline{N}_{L}(T)} = \frac{1}{L} \int_{0}^{T} \lambda(t)dt
		$}. This proves that
	\begin{equation}
		\left|\widehat{\lambda}(t)-\widehat{\lambda}(u_{j})\right| = 
		\mathcal{O}\left(\frac{\Delta(L)}{h^{2}}\right)
		\enspace \text(a.s.)
		\label{eqn:d5}
	\end{equation}
	uniformly in $t\in T_{\epsilon}$. Concerning the term $A_{2}$ in~\eqref{eqn:d3} we 
	can use~\eqref{eqn:4.15}. Then, we obtain
	\begin{equation*}
		\begin{split}
			& \E{\widehat{\lambda}(t)} - \E{\widehat{\lambda}(u_{j})} \\
			& = \int_{0}^{T} \left[
			K_{h}(t-s) - K_{h}(u_{j}-s)
			\right] \lambda(s) ds
		\end{split}
		\,.
	\end{equation*}
	This gives
	\begin{equation}
		A_{2} = \mathcal{O}\left(\frac{\Delta(L)}{h^{2}}\right)
		\label{eqn:d6} \,.
	\end{equation}
	Hence, we have shown that the terms $A_{1}$ and $A_{2}$ are of order \mbox{$
		\mathcal{O}\left(\frac{\Delta(L)}{h^{2}}\right)	
		$}, where $\Delta(L)$ is to be selected.
	
	Finally, let us consider the term $A_{3}$ in~\eqref{eqn:d3}. First we note that for 
	$\delta>0$
	\begin{equation}
		\begin{split}
			& \mathbf{P}\left(
			\max_{1\leq j\leq q(L)} \left|
			\widehat{\lambda}(u_{j}) - \E{\widehat{\lambda}(u_{j})}
			\right| \geq \delta
			\right) \\
			& \leq q(L) \sup_{t\in T_{\epsilon}} \mathbf{P}\bigg(
			\left|
			\widehat{\lambda}(t) - \E{\widehat{\lambda}(t)}
			\right| \geq \delta
			\bigg)
		\end{split}
		\label{eqn:d7} \,.
	\end{equation}
	By virtue of~\eqref{eqn:4.14} and~\eqref{eqn:4.11} we have
	\begin{equation*}
		\widehat{\lambda}(t) - \E{\widehat{\lambda}(t)} = \int_{0}^{T}
		K_{h}(t-s) d\overline{M}_{L}(s)
		\,.
	\end{equation*}
	This,~\eqref{eqn:4.12} and Chebyshev inequality yield
	\begin{equation*}
		\mathbf{P}\left(
		\left|
		\widehat{\lambda}(t) - \E{\widehat{\lambda}(t)}
		\right| \geq \delta
		\right) \leq \frac{1}{L} \int_{0}^{T} K^{2}_{h}(t-s)\lambda(s)ds/\delta^{2}
		\,.
	\end{equation*}
	Note that the right-hand side of this inequality is of order $\mathcal{O}(1/Lh)$ 
	uniformly in $t\in T_{\epsilon}$. This,~\eqref{eqn:d7} and the fact that
	$q(L) = \mathcal{O}(1/\Delta(L))$
	lead to the following uniform bound
	\begin{equation*}
		\mathbf{P}\left(A_{3}\geq\delta\right) = \mathcal{O}\left(
		1/\Delta(L)Lh
		\right)
	\end{equation*}
	or equivalently \mbox{$
		A_{3} = \mathcal{O}_{P}\left(1/\sqrt{\Delta(L)Lh}\right)	
		$}. 
	Hence, balancing $A_{3} = \mathcal{O}_{P}\left(1/\sqrt{\Delta(L)Lh}\right)$
	versus $A_1, A_2 = \mathcal{O}(\Delta(L)/h^2)$ gives the choice 
	\mbox{$\Delta(L)=h/L^{1/3}$}. This yields the convergence in~\eqref{eqn:d1} if
	\begin{equation*}
		h(L) \to 0
		\enspace \text{and}
		\enspace Lh^{3}(L) \to \infty
		\enspace \text{as}
		\enspace L \to \infty
		\,.
	\end{equation*}
	The proof of Lemma~\ref{lemma:8} has been completed.
\end{proof}

\section*{Acknowledgment}
\noindent This work was supported by the Polish National Center of Science under
Grant DEC-2017/27/B/ST7/03082 and NSERC Grant 319732.



\begin{IEEEbiography}
	[{\includegraphics[width=1in,height=1.25in,clip,keepaspectratio]
		{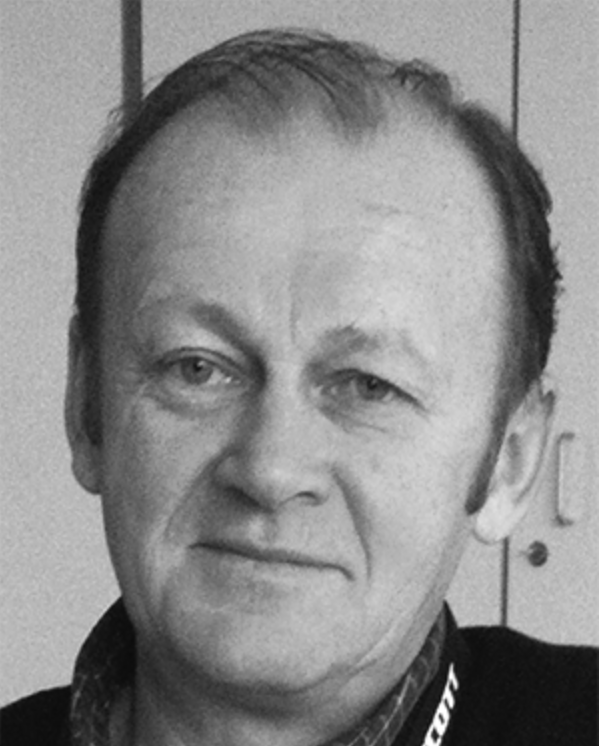}}]
	{Mirosław Pawlak} received
	the Ph.D. (under the supervision of Prof. Greblicki) and D.Sc. degrees in computer engineering from Wrocław University of Technology,
	Wrocław, Poland, in 1984 and 2006, respectively.
	He is currently a Professor with the Department of Electrical and Computer Engineering,
	University of Manitoba, Winnipeg, MB, Canada.
	He has held a number of visiting positions in
	North American, Australian, and European Universities. He was with the University of Ulm and University of Goettingen
	as an Alexander von Humboldt Foundation Fellow. Among his publications in these areas are the books Image Analysis by Moments (Wrocław
	Univ. Technol. Press, 2006), and Nonparametric System Identification
	(Cambridge Univ. Press, 2008), coauthored with Prof. Włodzimierz Greblicki. His research interests include statistical signal processing, machine learning, and nonparametric modeling.
	Dr. Pawlak has been an Associate Editor for the Journal of Pattern
	Recognition and Applications, Pattern Recognition, International Journal
	on Sampling Theory in Signal and Image Processing, Opuscula Mathematica and Statistics in Transition-New Series.
\end{IEEEbiography}

\begin{IEEEbiography}
	[{\includegraphics[width=1in,height=1.25in,clip,keepaspectratio]
		{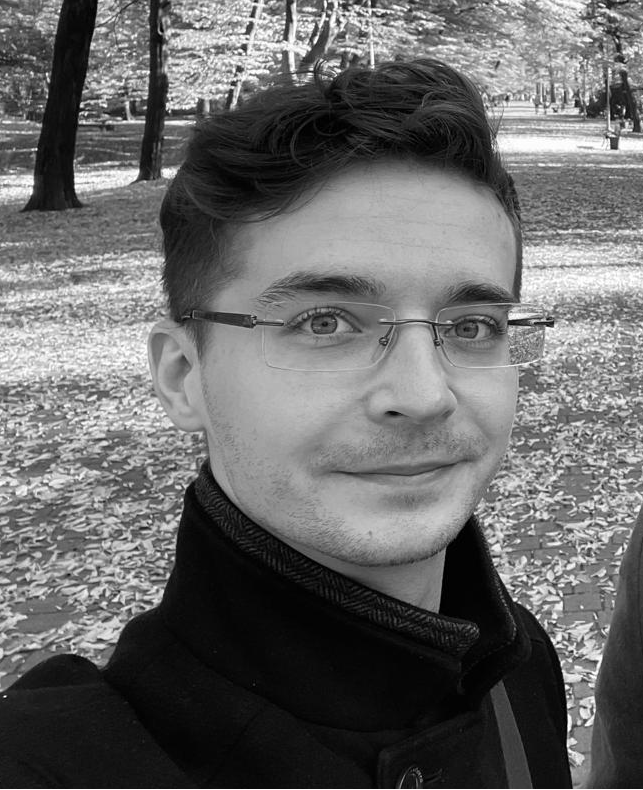}}]
	{Mateusz Pabian} received the M.Sc.
	degree in biomedical engineering from the AGH
	University of Science and Technology, Kraków,
	Poland, in 2019, where he is currently pursuing the
	Ph.D. degree at the Department of Measurement and Electronics.
	In 2017 he has completed an internship at NTT Communication Science Laboratories as
	part of the Vulcanus in Japan 2016-2017 Programme, which focused on experiment 
	design, signal acquisition and analysis of the auditory evoked potentials.
	From 2017 to 2022 he was a Machine Learning Researcher at Comarch Healthcare, where
	he has been involved in the development of computer-aided medical diagnostics
	systems.
	Since 2022 he has been working as Model Developer at UBS in Kraków, Poland.
	His research
	interests include machine learning and event-based systems, particularly
	spiking neural networks.
\end{IEEEbiography}

\begin{IEEEbiography}
	[{\includegraphics[width=1in,height=1.25in,clip,keepaspectratio]
		{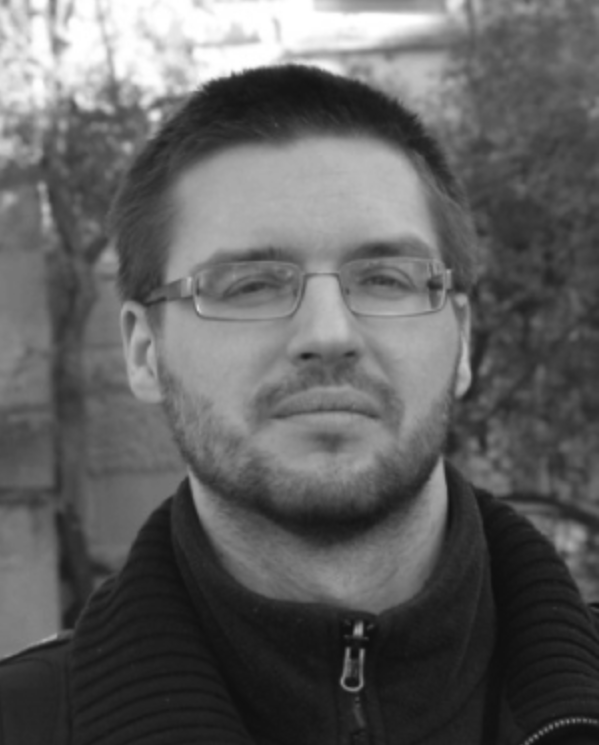}}]
	{Dominik Rzepka} received his M.Sc. and Ph.D. degree in
	electrical engineering from the AGH University of
	Science and Technology, Kraków, Poland, in 2009 and 2018 respectively. From 2007 to
	2011, he was with the Wireless Sensor and Control
	Networks Group, AGH University of Science and
	Technology, where he was involved in the design
	of the low-power algorithms for the processing of
	radio signals and software-defined radio. In 2011,
	he joined the Event-Based Control and Signal Processing Group at AGH
	University of Science and Technology, where he currently works on methods of signals reconstruction from event-triggered samples and on neuromorphic machine learning. He was
	a Visiting Student and Postdoc Researcher in the University of Manitoba, Winnipeg,
	Canada, from 2014 to 2023, and in The City College of New York, USA,
	in 2015. Since 2015, he is working as Signal Processing and Machine Learning Researcher in
	Comarch Healthcare and Fitech, developing algorithms for diagnostics and quality assurance systems. His research interests include
	signal processing and machine learning in biomedicine, wireless communication and industrial inspection, and event-based systems.
\end{IEEEbiography}

\end{document}